\RequirePackage{xcolor}
\documentclass[journal,twoside,web,letterpaper]{IEEEtran}
\usepackage[noadjust]{cite}
\usepackage{amsmath,amssymb,amsfonts}
\usepackage{algorithmic}
\usepackage{graphicx}
\usepackage{algorithm,algorithmic}
\usepackage{hyperref}
\usepackage{textcomp}

\usepackage[makeroom]{cancel}

\usepackage{amsthm}

\newtheorem{counterexample}{Counterexample}

\newtheoremstyle{ieeeconf}
  {0pt}   
  {0pt}   
  {\normalfont}  
  {\parindent}       
  {\itshape} 
  {:}         
  { } 
  {\thmname{#1} \thmnumber{#2}\thmnote{ (#3)}} 
\makeatletter
\renewenvironment{proof}[1][\proofname]{\par
  \pushQED{\qed}%
  \normalfont \topsep\z@
  \trivlist
  \item[\hskip2em
        \itshape
    #1\@addpunct{:}]\ignorespaces
}{%
  \popQED\endtrivlist\@endpefalse
}
\makeatletter

\theoremstyle{ieeeconf}

\newtheorem{theorem}{Theorem}
\newtheorem{definition}{Definition}
\newtheorem{proposition}{Proposition}

\newtheorem{corollary}{Corollary}
\newtheorem{assumption}{Assumption}
\newtheorem{remark}{Remark}

\def\BibTeX{{\rm B\kern-.05em{\sc i\kern-.025em b}\kern-.08em
    T\kern-.1667em\lower.7ex\hbox{E}\kern-.125emX}}
\begin{document}
\title{State-space fading memory}
\author{Gustave Bainier, Antoine Chaillet, Rodolphe Sepulchre~\IEEEmembership{Fellow,~IEEE}, Alessio Franci
\thanks{Submitted for review the 25/03/2026.}
\thanks{This work was supported by European Research Council under the Advanced ERC Grant Agreement SpikyControl 101054323.}
\thanks{G. Bainier and A. Franci are with the Dept. of Electrical Engineering and Computer Science, University of Liège, 10 allée de la Découverte, Liège 4000, Belgium. A. Franci is also with the WEL Research Institute, Wavre 1300, Belgium (e-mail: gustave.bainier@uliege.be; afranci@uliege.be).}
\thanks{A. Chaillet is with Université Paris-Saclay, CNRS, CentraleSupélec, Laboratoire des signaux et systèmes, Gif-sur-Yvette 91190, France.}
\thanks{R. Sepulchre is with KU Leuven, Department of Electrical Engineering (ESAT), STADIUS Center for Dynamical Systems, Signal Processing and Data Analytics, KasteelPark Arenberg 10, Leuven 3001, Belgium, and also with the Department of Engineering, University of Cambridge, Trumpington Street, Cambridge CB2 1PZ, United Kingdom.}}

\maketitle

\begin{abstract}
The fading-memory (FM) property captures the progressive loss of influence of past inputs on a system's current output and has originally been formalized by Boyd and Chua in an operator-theoretic framework. Despite its importance for systems approximation, reservoir computing, and recurrent neural networks, its connection with state-space notions of nonlinear stability, especially incremental ones, remains understudied. This paper introduces a state-space definition of FM. In state-space, FM can be interpreted as an extension of incremental input-to-output stability ($\delta$IOS) that explicitly incorporates a memory kernel upper-bounding the decay of past input differences. It is also closely related to Boyd and Chua's FM definition, with the sole difference of requiring uniform, instead of general, continuity of the memory functional with respect to an input-fading norm. We demonstrate that incremental input-to-state stability ($\delta$ISS) implies FM semi-globally for time-invariant systems under an equibounded input assumption. Notably, Boyd and Chua's approximation theorems apply to $\delta$ISS state-space models. As a closing application, we show that, under mild assumptions, the state-space model of a current-driven memristive device possesses the FM property.
\end{abstract}

\begin{IEEEkeywords}
Fading Memory, Incremental Input-to-State Stability, Incremental Input-to-Output Stability
\end{IEEEkeywords}

\section{Introduction}
\label{sec:introduction}

\subsection{Fading memory and the representation of stable systems}

Fading-Memory (FM), introduced by Stephen~P.~Boyd and Leon~O.~Chua, is an incremental stability property of input-output (I/O) systems, where the current outputs depend on past inputs, but the influence of older inputs gradually diminishes~\cite{Boyd1985}. In their Theorem~2, Boyd and Chua showed that a large class of FM operators can be approximated arbitrarily well by a Wiener model~\cite{Wiener1958}, i.e. by a linear time-invariant state-space model followed by a nonlinear static readout:
\begin{subequations} \label{eq:initialRepresentationOfFMApproximationTheorem}
\begin{align}
\dot{x}(t)&=Ax(t)+Bu(t), \\
y(t)&=g(x(t)),
\end{align}
\end{subequations}
with $u$ and $y$ the input and output signals respectively. The state matrix $A$ is Hurwitz and the static readout $g$ is smooth, ensuring both internal and external stability. In practice, the Wiener model can be obtained through identification~\cite{Schoukens2017}. This result effectively provides a canonical representation of FM operators. However, it applies in the operator-theoretic context and does not directly translate to nonlinear state-space systems theory.

For linear systems, the situation is simpler: in their Theorem~5, Boyd and Chua have established that Linear Time-Invariant (LTI) operators have FM if and only if they admit a convolution representation~\cite{Boyd1985}. This makes FM an almost automatic property of LTI systems in engineering practice. LTI systems with FM enjoy well-known I/O properties such as Bounded-Input Bounded-Output (BIBO), Converging-Input Converging-Output (CICO), and Periodic-Input Periodic-Output (PIPO).

In contrast, for nonlinear systems, the notions of representation, internal stability, and external stability are largely decoupled. For state-space nonlinear systems, classical asymptotic stability concerns only the internal dynamics (convergence to an equilibrium in the absence of input) and does not automatically imply well-behaved I/O behaviors, nor vice-versa. Stronger stability notions are therefore required. One such a notion is incremental input-to-state stability ($\delta$ISS), which implies internal asymptotic stability, and I/O properties analogous to linear systems: namely, BIBO, CICO and PIPO~\cite{Angeli2002ALA,Angeli2009,Kato2025}. Moreover, $\delta$ISS is preserved under I/O interconnections satisfying a small-gain condition, making it particularly powerful for analyzing large-scale nonlinear systems~\cite{Angeli2002ALA}. However, unlike FM systems, $\delta$ISS does not readily provide a canonical representation theorem similar to Boyd and Chua's approximation result.

A state-space stability notion that recovers Boyd and Chua's FM is needed in order to apply their approximation result to nonlinear state-space models. This paper introduces this missing state-space FM notion. Rather unsurprisingly, state-space FM closely relates to the notion of $\delta$ISS. We show that, under equibounded inputs, $\delta$ISS implies the state-space FM property semi-globally, offering a simple pathway to generalize Boyd and Chua's approximation result to the state-space setting: $\delta$ISS implies state-space FM, which implies operator-theoretic FM, which yields the representation theorem. Our broader objective is to establish state-space FM as a nonlinear analog of the stable convolution property in LTI system theory, and thus to motivate the theoretical development of this state-space notion.

\subsection{Context of the study}

FM can be observed in many nonlinear systems: rheological materials~\cite{Coleman1968}, memristive systems~\cite{Forni2025}, Recurrent Neural Networks (RNNs)~\cite{Matthews1993}, just to give a few examples. More generally any stable linear system with input saturation possesses FM. Our own original motivation for studying FM came from neuroscience, where the ion-channel kinetics of the Hodgkin-Huxley model act as voltage-driven memristors~\cite{Hodgkin1952,CHUA2012}, although the current-driven case is discussed in this paper, as it is the more standard one in engineering.

Given the widespread nature of the FM property, several attempts have been made to establish its abstract and rigorous formalization. To the best of the authors' knowledge, the first rigorous formulation of the FM concept was introduced by Bernard~D.~Coleman in a series of articles focused on applications to continuum mechanics, specifically the viscoelastic and rheological behavior of polymers~\cite{Coleman1960,Coleman1966,Coleman1968}. While Coleman’s work had a significant impact on continuum mechanics and thermodynamics~\cite{ilhav2019,day2013thermodynamics}, the concept of FM remained largely unnoticed in the control and signal-processing communities, until it was re-discovered and re-formalized in 1985 by Boyd and Chua for I/O operators~\cite{Boyd1985}. While Boyd and Chua's work focused on time-invariant operators and their approximation, one of its long-standing corollary contribution is the convolution theorem for LTI operators, as mentioned earlier in this introduction~\cite{Ortega2024}.

Although this early formalization of FM continues to be one of the most widely recognized, a series of alternative definitions were also developed. The first of these appears in the work of Jeff S. Shamma, who characterized FM in terms of \textit{finite-memory operators}~\cite{Shamma1991}, and subsequently introduced with Rongze Zhao a distinction between uniform and pointwise FM~\cite{Shamma1993}. A drawback of the operator-theoretic definitions of FM is that they typically require inputs to be defined over the entire negative time axis, infinitely into the past. This limitation motivated Irwin W. Sandberg to reformulate FM so that it applies to operators defined on both discrete-time and continuous-time positive domains $\mathbb{N}$~and~$[0,+\infty)$~\cite{Sandberg2001Z,Sandberg2002R,Sandberg2003}.

Research efforts have also been dedicated to study operators with the FM property. Notable contributions include the examination of linear FM operators~\cite{Partington1996}, system identification of FM systems~\cite{Dahleh1995}, and work clarifying the relationship between FM and the conventional concept of asymptotic stability for discrete-time systems~\cite{Zang2003}. Two research directions related to FM operators can be identified: the \textit{representation} question, i.e. which classes of systems can approximate any FM operator with arbitrary precision; and the \textit{universality} question, i.e. what class of systems can FM operators represent. Crucially, these questions helped to clarify key properties and limitations of several machine learning architectures. In Boyd and Chua’s original approximation theorem~\cite{Boyd1985}, it is shown, by virtue of the Stone-Weierstrass theorem, that (under mild assumptions) all time-invariant FM operators can be approximated by a finite Volterra series and by a state-space LTI system followed by a static nonlinear readout. This demonstrates that Wiener systems are universal for the class of FM systems. Further results have established that discrete-time FM systems can be approximated by RNNs~\cite{Matthews1993}, and bounds on the computational capacity of discrete-time systems with the FM property have also been derived~\cite{Dambre2012}. A clear articulation of FM's relationship to related concepts in RNNs can be found in this recent study~\cite{Ortega2025}. The FM property has also been investigated in the context of reservoir computing: since a RNN with FM holds recent input history in its current state while forgetting older inputs, training a simple readout of the state can be used to achieve real-time temporal computation, without having to train the recurrent dynamics itself. Theoretically, reservoir systems with FM have been established as universal~\cite{Gonon2021}, though subsequent analysis refined this view, asserting that the FM property was not strictly necessary for reservoir systems to achieve universality~\cite{Sugiura2024}.

Despite the recent development of a kernel-based definition of FM~\cite{Huo2024}, the operator-theoretic notion of FM has attracted little attention from control theorists~\cite{Sepulchre2021}. One notable exception is a thesis on the identification of neuronal system, which contains a proof that Slotine's contraction~\cite{LOHMILLER1998} implies the FM notion of Boyd and Chua, in the discrete-time setting, on compact sets and for equibounded inputs (Section~4.2.1 of~\cite{Burghi2019}). The control community has instead focused its interests on state-space definitions of FM~\cite{Karafyllis2018}. The FM state-space notion that is the closest to its operator theoretic counterpart is input-to-state stability (ISS), a property which imposes that the state of the system is asymptotically attracted to an equilibrium, up to a bounded error related to the largest past input value~\cite{Sontag2008}. In ISS systems, inputs that converge to zero will typically still steer the system toward its equilibrium, as past input values are progressively forgotten over time. This explicit FM property of ISS systems was first introduced in~\cite{Praly1996} using the notion of exp-ISS, and was later generalized to input-to-state dynamical stability (ISDS)~\cite{Grune2002}. However, the exact link between these notions and Boyd and Chua's original FM definition remains elusive. In particular, although the time-discount is explicit, neither exp-ISS nor ISDS captures the inherently \textit{incremental} nature of Boyd and Chua's FM, which is characterized by how mismatches in past inputs propagate to mismatches in current outputs. Some time-discounted incremental stability notions have recently been introduced in the control literature~\cite{Allan2021,Schiller2023}, but they are input-output-to-state, and were developed to study nonlinear detectability, a motivation unconnected to FM. We argue that the incremental state-space notion of FM is missing from the literature, and that incremental input-to-output stability ($\delta$IOS) and $\delta$ISS are much more natural candidates to obtain a state-space equivalent of Boyd and Chua's FM~\cite{Sontag1989}.

\subsection{Contributions}

After reviewing some preliminary definitions in Section~\ref{sec:notations}, the new state-space definition of FM is proposed in Section~\ref{sec:FM}. This new definition is shown to be closely related to both Boyd and Chua's original operator-theoretic definition, and to the $\delta$IOS property of state-space models. It is in fact the $\delta$IOS property with the addition of an explicit \textit{memory kernel} leading to the decay of past input values. The new definition is first motivated through the example of a memristive device in Section~\ref{ssec:Motivation}, and then properly introduced in Section~\ref{ssec:FMDef}. Some elementary results are derived from the existence of the memory kernel in Section~\ref{ssec:FMProp}.

The connection between our definition and Boyd and Chua's definition in terms of uniform and general continuity is presented in Section~\ref{sec:BoydFM}.

Is the introduced state-space FM equivalent to the $\delta$ISS property, and can Boyd and Chua's approximation theorem be applied to $\delta$ISS systems? Section~\ref{sec:deltaISSFM} is dedicated to answering these two questions. After introducing the $\delta$ISS property in Section~\ref{ssec:recapdeltaISS}, it is shown that $\delta$ISS implies input-to-state FM on all compact sets of initial conditions for time-invariant systems with equibounded inputs in Section~\ref{ssec:deltaISSFMyes} (thus, the input-to-state FM holds semi-globally). The representation theorem obtained by applying Boyd and Chua's approximation results to $\delta$ISS state-space models is briefly stated in Section~\ref{ssec:deltaISSapproximation}.

Finally, we also exhibit simple conditions for the state-space model of a current-driven memristive device to have FM in Section~\ref{sec:Mem}. Namely, that the current remains bounded, that the internal dynamic of the resistance is $\delta$ISS, and that the memristance value is Lipschitz-continuous with respect to its internal state. 

Section~\ref{sec:end} briefly concludes the paper and offers several perspectives for future works. 

Some settings where $\delta$ISS does not imply input-to-state FM are discussed in Appendix~\ref{app:deltaISSFMno}.

\section{Notations}\label{sec:notations}

A function $\alpha : [0,+\infty) \rightarrow [0,+\infty)$ is of class $\mathcal{K}$ if it is continuous, strictly increasing, and $\alpha(0)=0$. Moreover, it is of class $\mathcal{K}_{\infty}$ if it is of class $\mathcal{K}$ and $\lim_{r \to \infty} \alpha(r) = +\infty$.

A function $\beta : [0,+\infty) \times [0,+\infty) \rightarrow [0,+\infty)$ is of class $\mathcal{KL}$ if it is continuous, and such that for each $s\in [0,+\infty)$, $\beta(\cdot{},s) \in \mathcal{K}$, and for each $r\in [0,+\infty)$,  $\beta(r,\cdot{})$ is nonincreasing on $[0,+\infty)$ with $\lim_{s \to+\infty} \beta(r,s) = 0$.

A set of function $\mathcal{F}$ is equibounded if all its functions are uniformly bounded, i.e. there exists $M > 0$ such that $\lVert f(x) \rVert \le M$ for all $f \in \mathcal{F}$.

\section{FM for state-space models} \label{sec:FM}

\subsection{Motivating example: a current-driven memristor}\label{ssec:Motivation}

A memristor is a two-terminal electronic device whose resistance, or memristance, $M>0$ varies over time, typically depending on the history of current $I$ that has flowed through it, providing its behavior with a built-in memory of past activity~\cite{Chua1976}. For a memristor, Ohm's law is typically given in state-space representation and reads:
\begin{subequations} \label{eq:Memristor}
    \begin{align}
         U(t)&=M(x(t))I(t), \label{eq:Memristor1} \\
         \dot{x}(t)&=f(x(t),I(t)), \label{eq:Memristor2}
    \end{align}
\end{subequations}
where $U(t)$ stands for the voltage across the memristor. Equation~\eqref{eq:Memristor1} resembles the standard $U = RI$ relationship of a Ohmic resistor, except that the resistance $R$ is replaced by $M$, a memristance function, which depends on the internal state $x$ of the memristor. The memory effect of the memristance arises from the dynamics of the internal state variable~\eqref{eq:Memristor2}. In many memristive devices, the influence of past activity gradually diminishes, and they can serve as a canonical example of a FM system~\cite{Forni2025}.

While originally, Boyd an Chua defined FM through the continuity of the system operator with respect to a fading norm~\cite{Boyd1985}, our formalization highlights the system's retained memory by relying on its \textit{input-to-output contraction} properties. To gain some intuition, consider two identical memristive devices driven by two input currents $I_a$ and $I_b$ that differ before some time $t^*$ but become identical thereafter. Because the two devices retain a memory of past inputs, their corresponding voltages $U_a$ and $U_b$ still differ for $t\geq t^*$. However, as time progresses, the influence of the earlier discrepancies between $I_a$ and $I_b$ fades away, and the two voltages gradually converge toward each other at a rate characterizing the temporal scale of the system's memory. The paper proposes this systematic \textit{output contraction} in response to \textit{input contraction} as an alternative way to formalize FM. This definition will ultimately be shown to be closely related to that of Boyd and Chua, thus enabling the use of their original approximation theorem~\cite{Boyd1985}. In particular, if the memristive device has FM, then its model can can be approximated by a state-space LTI system with a static nonlinear readout, effectively removing the nonlinearity from the dynamical equation~\eqref{eq:Memristor2} of the original state-space model. Practical conditions for this FM property to hold on a memristive device are given at the end of the paper (Section~\ref{sec:Mem}).

\subsection{Definition} \label{ssec:FMDef}

Consider the nonlinear state-space model
\begin{subequations} \label{eq:StateSpace}
    \begin{align}
         \dot{x}(t)&=f(x(t),u(t)), \label{eq:StateSpace1} \\
         y(t)&=g(x(t),u(t)),  \label{eq:StateSpace2}
    \end{align}
\end{subequations}
where $x(t) \in \mathbb{R}^{n_x}$ is the state of the system, $u \in \mathcal{U}^+ \subset L_{\infty}([0,+\infty),\mathbb{R}^{n_u})$ is its input (the $+$ superscript in $\mathcal{U}^+$ denotes the positive-time domain $[0,+\infty)$), and $y(t) \in \mathbb{R}^{n_y}$ is its output. It is assumed that \eqref{eq:StateSpace1} is forward complete, meaning for all initial condition $x(0) \in \mathcal{X}_0 \subseteq \mathbb{R}^{n_x}$ and input $u \in \mathcal{U}^+$, the state trajectory $x$ is uniquely defined by \eqref{eq:StateSpace1} on the maximal interval of existence $t\in[0,+\infty)$, and moreover that $g$ is smooth in both arguments.

Intuitively, we propose that \eqref{eq:StateSpace} has the FM property if:
\begin{enumerate}
  \item The system is incrementally asymptotically output stable for all input signals $u\in \mathcal{U}^+$. This means that for any mismatch of initial conditions $\|x_a(0)-x_b(0)\|$, the trajectories of the two outputs, $y_a$ and $y_b$, considered with a common input signal $u\in \mathcal{U}^+$, converge toward each other uniformly across all initial conditions.
  \item The system is incrementally input-to-output stable (i.e. $\delta$IOS). This means that for two distinct input signals $u_a$ and $u_b \in \mathcal{U}^+$, $y_a$ and $y_b$ still converge to each other, but only up to a mismatch that scales with the maximum past mismatch between $u_a$ and $u_b$.
  \item The system's output gradually forgets its past inputs according to a uniform memory mechanism. This means that the mismatch between $y_a$ and $y_b$ is computed by weighting past mismatches between $u_a$ and $u_b$ according to their temporal distance, with respect to a system-specific fading \textit{memory kernel}. In particular, if $u_a$ and $u_b$ converge asymptotically to one another, then so do $y_a$ and $y_b$.
\end{enumerate}
Although these three properties were presented sequentially, it should be noted that $3)$ implies $2)$ which in turn implies $1)$. Consequently, a formalization of $3)$ alone would offer a complete characterization of FM. The memory kernel described in point $3)$ is defined by a function $w:[0,+\infty) \to [0,1]$, with $w(\Delta t)$ the decay weight associated to the input mismatch between $u_a$ and $u_b$ after a time $\Delta t \geq 0$. To ensure that the system exhibits FM, the discrepancies in the distant past must have progressively less influence on the present output behavior, as formalized in the following definition.
\begin{definition}[Memory kernel] \label{def:memorykernel} A function $w:[0,+\infty) \to [0,1]$ is a memory kernel if it is continuous, nonincreasing, and satisfies
    \begin{equation} \label{eq:MemoryKernelLimit}
        \lim_{\Delta t \to +\infty} w(\Delta t) = 0.
    \end{equation}    
\end{definition}
The formal definition of our state-space FM property is now introduced:

\begin{definition}[State-space FM] \label{def:FM}
    The system \eqref{eq:StateSpace} has $w$-FM if there exists $\beta \in \mathcal{KL}$, $\gamma \in \mathcal{K}_{\infty}$, and a memory kernel $w$ satisfying Definition~\ref{def:memorykernel} such that for any two initial conditions, $x_a(0),x_b(0) \in \mathcal{X}_0$, considered with two input trajectories, $u_a, u_b\in \mathcal{U}^+ $, the corresponding outputs satisfy for all $t\geq 0$:
\begin{align} \label{eq:FM_definition}
    \| y_a(t)-y_b(t) \| &\leq  \;\beta\big(\| x_a(0)-x_b(0) \|,t\big) \\&+\gamma\big(\operatorname{ess\,sup}_{s\in[0,t]} w(t-s)\| u_a(s)-u_b(s)\| \big). \nonumber
\end{align}
\end{definition}
\begin{remark}
   A qualitatively equivalent definition of FM is given by:
   \begin{align} \label{eq:FM_definition_max}
        \| y_a(t)-y_b(t) \| &\leq  \; \max \Big(\tilde{\beta}\big(\| x_a(0)-x_b(0) \|,t\big) , \\& \tilde{\gamma}\big(\operatorname{ess\,sup}_{s\in[0,t]} w(t-s)\| u_a(s)-u_b(s)\| \big) \Big). \nonumber
   \end{align}
    Succinctly, \eqref{eq:FM_definition} implies \eqref{eq:FM_definition_max} with $(\tilde{\beta},\tilde{\gamma})=(2\beta,2\gamma)$, and \eqref{eq:FM_definition_max} implies \eqref{eq:FM_definition} with $(\beta,\gamma)=(\tilde{\beta},\tilde{\gamma})$.
\end{remark}
\begin{remark}
    $\mathcal{X}_0$ may need to be taken forward invariant to ensure trajectories remain in the region where the inequality applies.
\end{remark}

This is essentially a $\delta$IOS property with the addition of a input-fading mechanism.

\begin{definition}[Input-to-state FM] \label{def:ISSFM}
    The system \eqref{eq:StateSpace} has input-to-state $w$-FM if it has $w$-FM with respect to the output $y=x$.
\end{definition}

Similarly, this is essentially a $\delta$ISS property with the addition of a input-fading mechanism.

\subsection{Elementary results using the memory kernel}\label{ssec:FMProp}

Similarly to most comparison functions in control, the memory kernel $w(\Delta t)$ provides an \textit{upper-bound} on the influence of past inputs on the current output. It is therefore not uniquely defined. A memory kernel satisfying inequality \eqref{eq:FM_definition} offers an \textit{a priori conservative} characterization of the system's memory constraints: the system may, in fact, be more forgetful than the memory kernel suggests. This ensures that the continuity and nonincreasing requirements of Definition~\ref{def:memorykernel} are obtained without loss of generality. This idea is formalized in the following proposition and corollary:

\begin{proposition}[$w_1 \leq w_2 \Rightarrow$ ($w_1$-FM $\Rightarrow$ $w_2$-FM)] \label{prop:MemoryKernelComp} Let the system \eqref{eq:StateSpace} have $w_1$-FM. The system \eqref{eq:StateSpace} also has $w_2$-FM for any memory kernel $w_2:[0,+\infty) \to [0,1]$ such that for all $\Delta t \geq 0$, $w_1(\Delta t) \leq w_2(\Delta t)$.
\end{proposition}
\begin{proof}
    The proof is immediate by upper-bounding \eqref{eq:FM_definition} in Definition~\ref{def:FM}.
\end{proof}

\begin{corollary} \label{cor:MemoryKernelComp} If \eqref{eq:StateSpace} has $w_1$-FM, with $w_1:[0,+\infty) \to [0,1]$ only satisfying condition~\eqref{eq:MemoryKernelLimit}, then there exists a memory kernel $w_2$ satisfying Definition~\ref{def:memorykernel} such that the system \eqref{eq:StateSpace} has $w_2$-FM.
\end{corollary}
\begin{proof} From Proposition~\ref{prop:MemoryKernelComp}, it is sufficient to show that for all $w_1:[0,+\infty)\rightarrow [0,1]$ such that $\lim_{x\to +\infty }w_1(x)=0$, there exists a continuous and nonincreasing $w_2:[0,+\infty)\rightarrow [0,1]$ such that $\lim_{x\to +\infty }w_2(x)=0$ and $w_2\geq w_1$ pointwise. This is, in fact, a usual result for comparison functions, see Lemma~2 of \cite{Kellett2014}. One can define $\tilde{w}(x) \triangleq \sup_{y\geq x} w_1(x)$ and obtain such $w_2$ from $\tilde{w}$ through standard regularization.
\end{proof}

A common choice for the memory kernel is the exponentially decaying function $w(\Delta t) = e^{-\mu \Delta t}$, which arises naturally in exponentially stable linear systems, such as the standard low-pass filter $\tau \dot{y}(t)=-y(t)+u(t)$, for which the following FM inequality holds:
\begin{align}
    \|\Delta y(t)\| 
    \leq& e^{-t/\tau}\|\Delta y(0)\|\\
    &+ 2 \,\operatorname{ess\,sup}_{s\in[0,t]} e^{-(t-s)/2\tau} \| \Delta u(s)\|. \nonumber
\end{align}

Similar to the properties of $\delta$ISS systems (see Proposition~4.4 and Proposition~4.5 in~\cite{Angeli2002ALA}), for FM systems, converging inputs lead to converging outputs (CICO), and periodic inputs lead to periodic outputs (PIPO). The first property is, however, more general than that of $\delta$ISS systems, since it solely relies on the memory kernel and do not require specific state-space properties (such as time-invariant or delay-free dynamics).
\begin{proposition}[$\Delta u(t) \to 0 \Rightarrow \Delta y(t) \to 0$] \label{prop:CICO}
    If the FM inequality \eqref{eq:FM_definition} stands and if $u_a$ and $u_b$ converge asymptotically to one another, then $y_a$ and $y_b$ do as well.
\end{proposition}
\begin{proof}
    Let $u_a,u_b\in \mathcal{U}^+$ be such that $\lim_{t\to +\infty} \| u_a(t)-u_b(t)\|= 0$. Since the input signals are in $L_{\infty}$, there exists $M>0$ such that $\operatorname{ess\,sup}_{t \geq 0} \| u_a(t)-u_b(t) \| <M$. Moreover, for concision, let us introduce $h(t,s) \triangleq  w(t-s) \| u_a(s)-u_b(s)\|$. It is sufficient to show that
    \begin{equation} \label{eq:ConvergenceConditionToBeVerified}
        \lim_{t\to +\infty}\operatorname{ess\,sup}_{s\in[0,t]}h(t,s) = 0.
    \end{equation}
    Let $\varepsilon>0$. Since $\lim_{t\to +\infty} w(t)= 0$, there exists $t_1>0$ such that for all $t\geq t_1$, $w(t) \leq \varepsilon/(2M)$. Similarly, since $\lim_{t\to +\infty} \| u_a(t)-u_b(t)\|= 0$, there exists $t_2>0$ such that for all $t\geq t_2$, $\| u_a(t)-u_b(t)\| \leq \varepsilon/2$. The following inequality holds:
    \begin{align}
        \operatorname{ess\,sup}_{s\in[0,t]} h(t,s) \leq & \operatorname{ess\,sup}_{s\in[0,t/2]}h(t,s)  \\
        &+ \operatorname{ess\,sup}_{s\in[t/2,t]}h(t,s). \nonumber
    \end{align}
    On one hand, for all $t\geq 2 t_1$:
    \begin{align}
        \operatorname{ess\,sup}_{s\in[0,t/2]}h(t,s) \leq \operatorname{ess\,sup}_{s\in[0,t/2]} w(t-s) M \leq \frac{\varepsilon}{2}.
    \end{align}
    On the other hand, for all $t\geq 2 t_2$:
    \begin{align}
        \operatorname{ess\,sup}_{s\in[t/2,t]}h(t,s) \leq \operatorname{ess\,sup}_{s\in[t/2,t]} \| u_a(s)-u_b(s)\| \leq \frac{\varepsilon}{2}.
    \end{align}
    Overall, if $t\geq 2 \max(t_1,t_2)$, then $\operatorname{ess\,sup}_{s\in[0,t]}h(t,s) \leq  \varepsilon$, which demonstrates \eqref{eq:ConvergenceConditionToBeVerified}.
\end{proof}

\begin{proposition}[Entrainment by periodic input] \label{prop:PIPO}
    Let the system~\eqref{eq:StateSpace} be bounded-input bounded-state (BIBS) and FM (Definition~\ref{def:FM}). If the input $u$ is $T$-periodic, then the output $y$ converges to a $T$-periodic trajectory, independently of the initial conditions.
\end{proposition}
\begin{proof}
    Let $u\in \mathcal{U}^+$. Since the input signal is in $L_{\infty}$, there exists $M_u>0$ such that $\operatorname{ess\,sup}_{t \geq 0} \| u(t) \| <M_u$. Now consider the trajectory $x$ of initial condition $x(0)\in \mathcal{X}_0$. The BIBO property implies that there exists $M_x>0$ and $t^*>0$ such that for all $t\geq t^*$, $\| x(t) \| < M_x$. Continuity of $g$ then implies $y \in L_{\infty}([0,+\infty),\mathbb{R}^{n_y})$. Let $t_1=t^*+k_1T+\delta$ and $t_2=t^*+k_2T+\delta$ with $k_1,k_2\in\mathbb{N}$ and $\delta \in[0,T)$. Let us assume without loss of generality that $k_1 \leq k_2$. An immediate application of the FM inequality~\eqref{eq:FM_definition} to the trajectories obtained with:
        \begin{align}
            &\begin{cases}
                u_{a,1}(\cdot) \triangleq u(\cdot+t^*) \\
            x_{a,1}(0)\triangleq x(t^*)
            \end{cases}\\
            &\mbox{and:} \nonumber\\
            &\begin{cases}
                u_{b,1}(\cdot) \triangleq u(\cdot+t^*+(k_2-k_1)T)\; [= u(\cdot+t^*)]\\
            x_{b,1}(0)\triangleq x(t^*+(k_2-k_1)T)
            \end{cases}
        \end{align}
    yields:
    \begin{align}
        &\sup_{\delta \in [0,T)}\|y(t^*+k_1T+\delta)-y(t^*+k_2T+\delta) \| \\
         & \leq\,\sup_{\delta \in [0,T)}\beta\big(\| x_{a,1}(0)-x_{b,1}(0) \|,k_1T+\delta\big)  \nonumber\\
        & +\gamma\big(\operatorname{ess\,sup}_{s\in[0,k_1T+\delta]} w(t_1-s) \| u_{a,1}(s)-u_{a,2}(s)\| \big) \nonumber \\
         &\leq\,\sup_{\delta \in [0,T)}\beta\big(\| x(t^*)-x(t^*+(k_2-k_1)T) \|,k_1T+\delta\big)   \nonumber\\
        & +\gamma\big(\operatorname{ess\,sup}_{s\in[t^*,t_1]} w(t_1-s) \underbrace{\| u(s)-u(s+(k_2-k_1)T)\|}_{=0} \big) \nonumber \\
        &\leq\, \beta\big(2M_x,k_1T\big) \to_{k_1 \to +\infty} 0.
    \end{align}
   Therefore, given that the output $y \in L_{\infty}([0,+\infty),\mathbb{R}^{n_y})$, the sequence $\left(y(t^*+kT+\cdot )\right)_{k \in \mathbb{N}}$ is a Cauchy sequence in the complete metric space $L_{\infty}([0,T), \mathbb{R}^{n_y})$, so it converges to $y^*\in L_{\infty}([0,T), \mathbb{R}^{n_y})$. Now, let us show the uniqueness of $y^*$ with respect to the initial condition by considering the trajectories associated with $x_{a,2}(0)$ and $x_{b,2}(0) \in \mathcal{X}_0$. Again, the FM inequality~\eqref{eq:FM_definition} yields for all $k\in \mathbb{N}$:
    \begin{align}
        &\sup_{\delta \in [0,T)}\|y_{a,2}(t^*+kT+\delta)-y_{b,2}(t^*+kT+\delta) \| \\
         &\leq\,\beta\big(\|x_{a,2}(t^*)-x_{b,2}(t^*)  \|,kT\big) \to_{k \to +\infty} 0.   \nonumber
    \end{align}
    This shows that all output sequences $\left(y(t^*+kT+\cdot )\right)_{k \in \mathbb{N}}$ converge to the same limit $y^*$, regardless of the initial conditions. Finally, considering the $T$-periodic function $y_p\in  L_{\infty}([0,+\infty),\mathbb{R}^{n_y})$ defined by $y_p(t)\triangleq y^*(t-\lfloor t/T \rfloor T)$, this implies $\lim_{t\to +\infty} \| y(t^*+t)-y_p(t) \| = 0$.
\end{proof}

\begin{remark}
    Proposition~\ref{prop:PIPO} can also be established from $\delta$IOS alone rather than state-space FM.
\end{remark}

\begin{remark}
    If the system has input-to-state FM (Definition~\ref{def:ISSFM}), Proposition~\ref{prop:CICO} and Proposition~\ref{prop:PIPO} are exactly Proposition~4.5 and Proposition~4.4 of~\cite{Angeli2002ALA}. The proof of Proposition~\ref{prop:CICO} is in fact relatively similar to that of Proposition~4.4 of~\cite{Angeli2002ALA}. An interesting corollary of Proposition~\ref{prop:PIPO} in this case is that, for constant inputs, all solutions converge to an equilibrium.
\end{remark}

We conclude this section with an evident, yet compelling, application of the memory kernel: it enables the derivation of explicit time-varying constraints on $\|\Delta u(t)\|$ to achieve a prescribed $\|\Delta y(t)\|$ at a given time horizon $t^*$. This slightly refines the ultimate-bound property of $\delta$IOS by replacing the static bounds on past perturbations with less conservative time-varying ones.
\begin{proposition} \label{prop:bounds}
    Let $t^*, r>0$, and let the system \eqref{eq:StateSpace} have $w$-FM (Definition~\ref{def:FM}) with $w>0$ pointwise. If two inputs $u_a,u_b\in \mathcal{U}^+$ satisfy the following constraint:
    \begin{equation}
        \| u_a(t)-u_b(t)\| \leq  \begin{cases} \gamma^{-1}(r)/w(t^*-t) & \mbox{ if } t\in[0,t^*]\\
        \gamma^{-1}(r) &\mbox{ if } t>t^* \end{cases}
    \end{equation}
    then the mismatch between $y_a$ and $y_b$ at time $t\geq t^*$ satisfies:
    \begin{equation}
        \| y_a(t)-y_b(t)\| \leq \beta\big(\| x_a(0)-x_b(0) \|,t\big)  +r.
    \end{equation}
\end{proposition}
\begin{proof}
    Again, the proof is immediate by upper-bounding \eqref{eq:FM_definition} in Definition~\ref{def:FM}, with $t\geq t^*$. Succinctly:
    \begin{align}
        &\operatorname{ess\,sup}_{s\in[0,t]} w(t-s)\|\Delta u(s)\| \\
        &\leq \max\left( \operatorname{ess\,sup}_{s\in[0,t^*]} w(t-s)\|\Delta u(s)\|, \right. \nonumber\\
        &\;\;\left. \operatorname{ess\,sup}_{s\in[t^*,t]}\|\Delta u(s)\|\right) \nonumber \\
        &\leq \max\left( \operatorname{ess\,sup}_{s\in[0,t^*]} \frac{w(t-s)}{w(t^*-s)}\gamma^{-1}(r), \gamma^{-1}(r)\right).
    \end{align}
    Since $w$ is taken nonincreasing, $w(t-s)/ w(t^*-s)\leq 1$ for all $t\geq t^*$ and $s\in[0,t^*]$, so:
    \begin{equation}
        \operatorname{ess\,sup}_{s\in[0,t]} w(t-s)\|\Delta u(s)\|  \leq \gamma^{-1}(r),
    \end{equation}
    which concludes the proof.
\end{proof}

\section{Connection to Boyd and Chua's FM} \label{sec:BoydFM}

Boyd and Chua's definition of FM was developed for time-invariant and causal I/O operators $G:\mathcal{U} \to \mathcal{Y}$ with $\mathcal{U} \subset L_{\infty}(\mathbb{R},\mathbb{R}^{n_u})$ and  $\mathcal{Y} \subset L_{\infty}(\mathbb{R},\mathbb{R}^{n_y})$. In this context, the inputs and outputs are required to be defined over the entire time axis (including negative time), as recalled in the following definitions of time-invariance and causality.

\begin{definition}[Time invariance] Let us denote by $T_{\tau}:\mathcal{L}_{\infty} \to \mathcal{L}_{\infty}$ the time translation defined by $T_{\tau} f (t) \triangleq f(t-\tau)$. The operator $G:\mathcal{U} \to \mathcal{Y}$ is time-invariant if for all $u\in \mathcal{U}$ and $\tau\in\mathbb{R}$:
    \begin{align}
        &T_{\tau}u \in \mathcal{U},  & [\mbox{input set invariance}] \label{eq:inputsetinvariance}\\
        &G(T_{\tau}u) = T_{\tau}G(u).  & [\mbox{operator equivariance}] \label{eq:operatorinvariance}
    \end{align}
\end{definition}

\begin{remark}
    The operator $G:\mathcal{U} \to \mathcal{Y}$ is equivariant with respect to the group of translation of the real line, so the time-invariance property is in fact a time-equivariance property. However, this terminology is not commonly used by the literature~\cite{Donchev2025}.
\end{remark}

\begin{definition}[Causality] The operator $G:\mathcal{U} \to \mathcal{Y}$ is time causal if for all $u_a,u_b \in \mathcal{U}$ and $t_0\in \mathbb{R}$:
    \begin{equation}
        u_a(t) = u_b(t), \forall t \leq t_0 \Rightarrow  G(u_a)(t) = G(u_b)(t) , \forall t \leq t_0. \label{eq:cond1timecausality}
    \end{equation}
\end{definition}

\begin{remark}
    Intuitively, this definition states that at each instant $t_0$, the output is solely determined by the past of the input, not its future.
\end{remark}

It is relatively well-known that this causality condition can be further simplified in the context of time-invariance.

\begin{proposition}[Causality under time-invariance~\cite{Boyd1985}] A time-invariant operator $G:\mathcal{U} \to \mathcal{Y}$ is time causal if for all $u_a,u_b \in \mathcal{U}$:
\begin{equation}
        u_a(t) = u_b(t), \forall t \leq 0 \Rightarrow  G(u_a)(t) = G(u_b)( t) , \forall t \leq 0. \label{eq:cond2timecausality}
\end{equation}
\end{proposition}

The properties of time-invariance and causality mean that $G$ can solely be defined and studied through a \textit{memory functional}, which maps its input before time $t=0$ to its output at time  $t=0$~\cite{Boyd1985}. We introduce the set $\mathcal{U}^- \triangleq \{u^-,  u \in \mathcal{U} \}$, where $u^-$ stands for the restriction of the function $u$ to negative time. $\mathcal{U}^+$ and $u^+$ are defined likewise for positive time. A definition for the memory functional of an operator is defined thereafter.

\begin{definition}[Memory functional] Given the operator $G:\mathcal{U} \to \mathcal{Y}$, the map $F_G : \mathcal{U}^- \to\mathbb{R}^{n_y}$ is a memory functional of $G$ if for all $u^-\in \mathcal{U}^-$, there exists $\tilde{u} \in \mathcal{U}$ such that $u^-=\tilde{u}^-$ and $F_G(u^-) =G(\tilde{u})(0)$.
\end{definition}

At this stage, the definition of a memory functional remains underspecified: the memory functional of an operator is not uniquely defined, nor is the operator uniquely determined by its memory functional. While such memory functionals are not useful in themselves, they serve to demonstrate that time-invariance and causality constitute the minimal assumptions required to guarantee a convenient correspondence between an operator and its associated memory functional. This observation is purely rhetorical in nature.

\begin{proposition} \label{prop:unique} If $G:\mathcal{U} \to \mathcal{Y}$ is causal, then $G$ defines a unique memory functional $F_G : \mathcal{U}^- \to\mathbb{R}^{n_y}$. Conversely, if $G:\mathcal{U} \to \mathcal{Y}$ is causal and time-invariant, then the memory functional $F_G : \mathcal{U}^- \to\mathbb{R}^{n_y}$ uniquely defines $G$.
\end{proposition} 
\begin{proof} First, let us show that a causal $G$ uniquely defines $F_G$. Let $F_G^a$ and $F_G^b$ be two memory functionals associated with $G$. Let $u^-\in \mathcal{U}^{-}$ and take $u_a,u_b\in \mathcal{U}$ such that $u^-=u_a^-=u_b^-$, $F_G^a(u^-)=G(u_a)(0)$ and $F_G^b(u^-)=G(u_b)(0)$. By definition of causality \eqref{eq:cond1timecausality}, $G(u_a)(t)=G(u_b)(t)$ for all $t\leq 0$, so in particular, $G(u_a)(0)=G(u_b)(0)$, which implies that $F_G^a(u_a^-) = F_G^b(u_b^-)$, and thus $F_G^a(u^-)=F_G^b(u^-)$ for all $u^-\in \mathcal{U}^{-}$, i.e. $F_G^a=F_G^b$.

Conversely, let us show that if $G$ is causal and time-invariant, then $F_G$ uniquely defines $G$. Let us take $F_{G_a}$ and $F_{G_b}$ the unique memory functionals associated with two time invariant and causal operators $G_a$ and $G_b$. Assuming $F_{G_a}=F_{G_b}$, we have $G_a(u)(0)= G_b(u)(0)$ for all $u\in \mathcal{U}$. By the input set invariance \eqref{eq:inputsetinvariance}, if $u \in \mathcal{U}$, then for all $\tau \in \mathbb{R}$, we have $T_{-\tau}u \in \mathcal{U}$, hence $G_a(T_{-\tau}u)(0)= G_b(T_{-\tau}u)(0)$, which is to say, by the operator equivariance \eqref{eq:operatorinvariance}, $G_a(u)(\tau)= G_b(u)(\tau)$ for all $u\in \mathcal{U}$ and $\tau \in \mathbb{R}$, and thus $G_a=G_b$.
\end{proof}

Boyd and Chua's definition of FM relies on this memory functional: an operator $G$ has FM if its memory functional $F_G$ is continuous with respect to the $w$-fading input norm. Boyd and Chua's definition is stated hereafter.

\begin{definition}[FM operator~\cite{Boyd1985}] \label{def:OperatorFM} Let $G:\mathcal{U} \to \mathcal{Y}$ be a causal and time-invariant operator. $G$ has FM if there exists a memory kernel $w$ satisfying Definition~\ref{def:memorykernel} such that for all $u_a,u_b\in \mathcal{U}^-$, the memory functional $F_G : \mathcal{U}^- \to\mathbb{R}^{n_y}$ satisfies:
    \begin{align}
       \forall \varepsilon >0, \exists \delta >0,\;\; &\operatorname{ess\,sup}_{s\leq 0} w(-s)\| u_a(s)-u_b(s)\| <\delta \\
       &\Rightarrow \| F_G(u_a)-F_G(u_b) \| < \varepsilon. \nonumber
    \end{align}
\end{definition}

The state-space definition of FM introduced in Section~\ref{ssec:FMDef} is slightly more contrived than that of Boyd and Chua. Succinctly, the operator theoretic equivalent to our state-space FM definition consists in imposing \textit{uniform continuity} to the FM operator rather than solely continuity. Although imposing uniform continuity is stronger than continuity, the two notions of FM become identical for \textit{well-behaved inputs}, by which we mean inputs that are both Lipschitz-continuous and equibounded. Thus, we argue that the two notions are fundamentally equivalent in practice.

\begin{definition}[Uniform FM operator] \label{def:OperatorUniformFM} Let $G:\mathcal{U} \to \mathcal{Y}$ be a causal and time-invariant operator. $G$ has uniform FM if there exists $\gamma\in \mathcal{K}_{\infty}$ and a memory kernel $w$ satisfying Definition~\ref{def:memorykernel} such that for all $u_a,u_b\in \mathcal{U}^-$, the memory functional $F_G : \mathcal{U}^- \to\mathbb{R}^{n_y}$ satisfies:
\begin{align} \label{eq:FM_inequality}
 \| F_G(u_a)-F_G(u_b) \| \leq \gamma\left (\operatorname{ess\,sup}_{s\leq 0} w(-s) \| u_a(s)-u_b(s)\| \right).
\end{align}
\end{definition}
\begin{remark}
    In this context, $\gamma$ is the modulus of continuity of the memory functional $F_G$.
\end{remark}

\begin{remark}
    Similarly to the memory kernel discussed in Section~\ref{ssec:FMProp}, the memory kernel introduced in the two definitions hereabove is taken continuous and nonincreasing without loss of generality.
\end{remark}

\begin{proposition}[Well behaved $\mathcal{U}$: (FM operator $\Leftrightarrow$ Uniform FM operator)] \label{prop:FM_UniformFM} Let $G:\mathcal{U} \to \mathcal{Y}$ be a causal and time-invariant operator. If there exists $M,K>0$ such that:
\begin{equation}
    \mathcal{U} \subseteq \{u : \|u(t)\|\leq M, \|u(t)-u(s)\| \leq K |t-s|, \forall t,s \in \mathbb{R}\},
\end{equation}
then $G$ has FM {(Definition~\ref{def:OperatorFM})} if and only if it has uniform FM {(Definition~\ref{def:OperatorUniformFM})}.
\end{proposition}
\begin{proof}
    Following the proof of Lemma~1 in Section~4 of~\cite{Boyd1985}, we can show that $\mathcal{U}^-$ is contained in a compact set with respect to the $w$-fading input norm. The equivalence between FM and uniform FM then follows from the Heine-Cantor theorem, which demonstrates the equivalence between continuity and uniform continuity {inside compact sets.}
\end{proof}
We now show the equivalence between the operator-theoretic notion of uniform FM {(Definition~\ref{def:OperatorUniformFM})} and the state-space notion of FM introduced earlier {(Definition~\ref{def:FM})}. To this aim, we need to associate a fixed initial condition $x_0 \in \mathcal{X}_0$ of the state-space model \eqref{eq:StateSpace} with a fixed negative-time input $u_0 \in \mathcal{U}^-$ of the I/O operator $G$. The formal definition of a $(x_0,u_0)$-realization is introduced hereafter.

\begin{definition}[$(x_0,u_0)$-realization] The causal and time-invariant operator $G:\mathcal{U} \to \mathcal{Y}$ is a $(x_0,u_0)$-realization of the state-space model \eqref{eq:StateSpace} if the negative-time input $u_0 \in \mathcal{U}^-$ satisfies:
\begin{itemize}
    \item for all $u \in \mathcal{U}^+$, the temporal concatenation $u_0 \diamond u$ is an element of $\mathcal{U}$;
    \item the output $y$ of \eqref{eq:StateSpace} considered at $x_0 \in \mathcal{X}_0$ and associated with $u \in \mathcal{U}^+$ satisfies $y(t) = G (u_0\diamond u)(t)$ for all  $t\geq 0$.
\end{itemize}
\end{definition}

\begin{remark}
   Note that $G (u_0\diamond u)(t) = F_G((T_{-t}(u_0 \diamond u))^-)$.
\end{remark}

The equivalence between the operator-theoretic notion of uniform FM and the state-space notion of FM can then be stated as follows:

\begin{theorem}[Uniform FM operator $\Leftrightarrow$ State-space FM] \label{th:equivOpeSS} Let $G:\mathcal{U} \to \mathcal{Y}$ be a $(x_0,u_0)$-realization of the state-space model \eqref{eq:StateSpace}. The operator $G$ has uniform $w$-FM {(Definition~\ref{def:OperatorUniformFM})} for all inputs of the form $T_{-t}(u_0 \diamond u)$, with $u\in\mathcal{U^+}$ and $t\geq 0$, if and only if the state-space model~\eqref{eq:StateSpace} has $w$-FM {(Definition~\ref{def:FM})} for $\mathcal{X}_0 = \{x_0\}$.
\end{theorem}
\begin{proof}
 Since for all $u_a, u_b \in \mathcal{U}^+$ and $t\geq 0$, we have:
  \begin{align}
 &\gamma\left (\operatorname{ess\,sup}_{s\leq 0} w(-s) \| T_{-t}(u_0 \diamond u_a)(s)-T_{-t}(u_0 \diamond u_b)(s)\| \right) \\
 &=  \gamma\left (\operatorname{ess\,sup}_{s\leq t} w(t-s)  \| (u_0 \diamond u_a)(s)-(u_0 \diamond u_b)(s)\|  \right)  \nonumber \\
 &=  \gamma\left (\operatorname{ess\,sup}_{s\in[0,t]} w(t-s) \| u_a(s)-u_b(s)\| \right),
 \end{align}
 and 
 \begin{align}
      &\| y_a(t)-y_b(t) \| = {\| G (u_0 \diamond u_a)(t) -  G (u_0 \diamond u_b)(t) \|}
      \\&=\| F_G((T_{-t}(u_0 \diamond u_a))^-)-F_G((T_{-t}(u_0 \diamond u_b))^-) \|, \nonumber
 \end{align}
the inequalities \eqref{eq:FM_inequality} and \eqref{eq:FM_definition} are equivalent.
\end{proof}
Since the equivalence is obtained at a fixed $x_0$, it does not yet justify the $\mathcal{KL}$-dissipation term $\beta$ in Definition~\ref{def:FM}. An operator-theoretic justification for this term is provided hereafter, although the strict equivalence between the two FM notions is lost. Succinctly, if the I/O operator $G$ is a $(x_0,u_0)$-realization of the state-space model~\eqref{eq:StateSpace} for a whole a set of initial conditions $x_0\in\mathcal{X}_0$, with an appropriate negative-time control $u_0 \in \mathcal{U}^-$ varying uniformly continuously with respect to the choice of $x_0$, then a $\mathcal{KL}$-dissipation term $\beta$ can be recovered solely from the operator-theoretic Definition~\ref{def:OperatorUniformFM} of FM.

\begin{assumption} \label{ass:operator_statespace2} Let $\mathcal{U}^-$ be equipped with the usual $L_{\infty}$ norm. We assume that there exists a uniformly continuous map $\psi:\mathcal{X}_0  \to \mathcal{U}^-$, with $\alpha\in \mathcal{K}_{\infty}$ its modulus of continuity, such that for all $x_0\in \mathcal{X}_0$, the causal and time-invariant operator $G:\mathcal{U} \to \mathcal{Y}$ is a $(x_0,\psi(x_0))$-realization of the state-space model \eqref{eq:StateSpace} considered at $x_0 \in \mathcal{X}_0$.
\end{assumption}

\begin{remark}
    In the case where $\mathcal{X}_0$ is compact, the Heine-Cantor theorem guarantees that solely continuity of $\psi$ is required.
\end{remark}

As an evident example, we can consider once again the standard low-pass filter $\tau \dot{y}(t)=-y(t)+u(t)$ and its I/O operator realization $G$ defined for all $u\in \mathcal{U}$ and $t \in \mathbb{R}$ by:
\begin{equation} \label{eq:SimpleOperatorExample}
    G(u)(t)\triangleq \frac{1}{\tau}\int_{-\infty}^t e^{-(t-s)/\tau}u(s) ds.
\end{equation}
Taking $\psi(x_0) \triangleq 1_{\leq 0}x_0$ over the domain $\mathcal{X}_0 = \mathbb{R}$, with $1_{\leq 0}$ the constant negative-time function of value $1$, we can easily verify that $\psi$ is uniformly continuous, and that for all $u\in \mathcal{U}^+$ and $t\geq 0$:
\begin{align}
    &G (\psi(x_0)\diamond u)(t)= e^{-t/\tau}x_0+ \frac{1}{\tau}\int_{0}^t e^{-(t-s)/\tau}u(s)ds=y(t),
\end{align}
so Assumption~\ref{ass:operator_statespace2} holds globally.

Given this assumption, the $\mathcal{KL}$-dissipation of initial conditions is recovered as follows:

\begin{proposition}[Uniform FM operator $\Rightarrow$ State-space FM] \label{prop:UniFMimpliesSSFM}Under Assumption~\ref{ass:operator_statespace2}, if $G$ has uniform $w$-FM {(Definition~\ref{def:OperatorUniformFM})} for all inputs in $\bigcup_{t\geq 0}T_{-t}(\psi(\mathcal{X}_0) \diamond \mathcal{U}^+)$, then the state-space model~\eqref{eq:StateSpace} has $w$-FM {(Definition~\ref{def:FM})} on $\mathcal{X}_0$.
\end{proposition}
\begin{proof}
    Let $x_{0,a},x_{0,b} \in \mathcal{X}_0$ and $u_a, u_b \in \mathcal{U}^+$. The outputs of the state-space model~\eqref{eq:StateSpace} associated with these initial conditions and inputs are denoted $y_a$ and $y_b$ respectively. Under Assumption~\ref{ass:operator_statespace2}, the following equality holds for all $t\geq 0$:
    \begin{align}
        &\| y_a(t)-y_b(t) \|=\\
     & \| F_G((T_{-t}( \psi(x_{0,a}) \diamond u_a))^-)-F_G((T_{-t}( \psi(x_{0,b})\diamond u_b))^-) \|,  \nonumber
    \end{align}
    thus, if $G$ has uniform FM, then the following inequality stands for all $t\geq 0$:
\begin{align}
     &\| y_a(t)-y_b(t) \| \\
     &\leq \gamma\left ( \max \left( \operatorname{ess\,sup}_{s\leq 0} w(t-s)  \| \psi(x_{0,a})(s)-\psi(x_{0,b})(s)\|, \right. \right.  \nonumber\\
     &  \left.\left.\operatorname{ess\,sup}_{s\in(0,t]} w(t-s)  \| u_a(s)-u_b(s)\| \right)\right).
\end{align}
By uniform continuity of $\psi$, the inequality $\operatorname{ess\,sup}_{s\leq 0}  \| \psi(x_{0,a})(s)-\psi(x_{0,b})(s)\| \leq  \alpha(  \| x_{0,a}- x_{0,b}\| )$ holds, hence, for all $t\geq 0$:
\begin{align}
    &\| y_a(t)-y_b(t) \| \\
     & \leq \max  \left(\gamma\left( \operatorname{ess\,sup}_{s\leq 0} w(t-s)  \alpha(  \| x_{0,a}- x_{0,b}\| ) \right), \right. \nonumber\\
     &  \left.\gamma\left( \operatorname{ess\,sup}_{s\in(0,t]} w(t-s)  \| u_a(s)-u_b(s)\|\right)\right) \nonumber \\
     &\leq \max  \left(\gamma\left( \alpha(  \| x_{0,a}- x_{0,b}\| ) \right), \right.\\
     &  \left.\gamma\left( \operatorname{ess\,sup}_{s\in(0,t]} w(t-s)  \| u_a(s)-u_b(s)\| \right)\right). \nonumber
\end{align}
Without loss of generality, the memory kernel $w$ is taken continuous and nonincreasing, so we recover the FM inequality~\eqref{eq:FM_definition_max}, with $\beta(r,t) \triangleq \gamma\left( w(t) \alpha(r)\right)$.
\end{proof}

The purpose of Proposition~\ref{prop:UniFMimpliesSSFM} is primarily rhetorical rather than practical. Although it may be difficult to identify a set $\mathcal{X}_0$ on which Assumption~\ref{ass:operator_statespace2} holds in practice, its (most likely) existence for a well-selected I/O operator realization of the state-space model~\eqref{eq:StateSpace} is sufficient to motivate the $\mathcal{KL}$-dissipation of initial conditions in Definition~\ref{def:FM}.

The converse, however, will turn out to be much more practical, since it allows to access the operator-theoretic approximation theorem of Boyd and Chua~\cite{Boyd1985} from the state-space definition of FM over a whole set of initial conditions. However, once again because of the $\mathcal{KL}$-dissipation of initial conditions, this converse result must be handled carefully. Namely, we need the composition of the inverse map $\psi^{-1}:\psi(\mathcal{X}_0) \to \mathcal{X}_0$ with the $\mathcal{KL}$-dissipation $\beta$ to itself satisfy a FM inequality.

\begin{assumption} \label{ass:operator_statespace3} We assume that there exists a map $\psi:\mathcal{X}_0  \to \mathcal{U}^-$ such that for all $x_0\in \mathcal{X}_0$, the causal and time-invariant operator $G:\mathcal{U} \to \mathcal{Y}$ is a $(x_0,\psi(x_0))$-realization of the state-space model \eqref{eq:StateSpace} considered at $x_0 \in \mathcal{X}_0$. Moreover we assume that there exist $\tilde{\gamma}\in \mathcal{K}_{\infty}$ and a memory kernel $\tilde{w}$ satisfying Definition~\ref{def:memorykernel} such that the following FM inequality holds for all $x_{0,a},x_{0,b} \in \mathcal{X}_0$ and $t\geq 0$:
\begin{align}
 &\beta(\| x_{0,a}-x_{0,b} \|,t) \\
 &\leq  \tilde{\gamma}\left (\operatorname{ess\,sup}_{s\leq 0}\tilde{w}(t-s) \| \psi(x_{0,a})(s)-\psi(x_{0,b})(s)\| \right). \nonumber
\end{align}
\end{assumption}

In the context of the low-pass filter $\tau \dot{y}(t)=-y(t)+u(t)$ realized by the operator~\eqref{eq:SimpleOperatorExample}, one can consider $\psi(x_0)(t) \triangleq 3 e^{t/2\tau}x_0/2$, and verify Assumption with $\beta(r,t) = e^{-t/\tau}r$, $\tilde{\gamma}(r) \triangleq 2r/3$ and $\tilde{w}(t)=e^{-t/2\tau}$. {Practical conditions for verifying this assumption in general nonlinear systems remain an open question.} 

\begin{proposition}[State-space FM $\Rightarrow$ Uniform FM operator] \label{prop:SSFMimpliesUniFM} Under Assumption~\ref{ass:operator_statespace3}, if the state-space model~\eqref{eq:StateSpace} has $w$-FM {(Definition~\ref{def:FM})} on $\mathcal{X}_0$, then $G$ has uniform FM {(Definition~\ref{def:OperatorUniformFM})} for all inputs in $\bigcup_{t\geq 0}T_{-t}(\psi(\mathcal{X}_0) \diamond \mathcal{U}^+)$.
\end{proposition}
\begin{proof} Let $x_{0,a},x_{0,b} \in \mathcal{X}_0$ and $u_a, u_b \in \mathcal{U}^+$. The outputs of the state-space model~\eqref{eq:StateSpace} associated with these initial conditions and inputs are denoted $y_a$ and $y_b$ respectively. Under Assumption~\ref{ass:operator_statespace2}, the following equality holds for all $t\geq 0$:
    \begin{align}
        &\| y_a(t)-y_b(t) \| = \\
     & \| F_G((T_{-t}( \psi(x_{0,a}) \diamond u_a))^-)-F_G((T_{-t}( \psi(x_{0,b})\diamond u_b))^-) \|,  \nonumber
    \end{align}
thus, if the state-space model~\eqref{eq:StateSpace} has $w$-FM, then the following inequality stands for all $t\geq 0$:
\begin{align}
     &\| y_a(t)-y_b(t) \| \\
     &\leq \max \left ( \beta(  \|x_{0,a}-x_{0,b}\|, t), \right.  \nonumber\\
     &  \left. \gamma\left( \operatorname{ess\,sup}_{s\in[0,t]} w(t-s)  \| u_a(s)-u_b(s)\| \right)\right).  \nonumber
\end{align}
If Assumption~\ref{ass:operator_statespace3} is verified, the first term can be upper-bounded as follows:
\begin{align}
 &\beta(\|x_{0,a}-x_{0,b}\|, t) \\
 &\leq  \tilde{\gamma} \left(\operatorname{ess\,sup}_{s\leq 0}\tilde{w}(t-s) \| \psi(x_{0,a})(s)-\psi(x_{0,b})(s)\| \right). \nonumber
\end{align}
Let us consider $\gamma_{ \tiny \mbox{tot}}\in \mathcal{K}_{\infty}$ and $w_{ \tiny \mbox{tot}}$ a memory kernel, with $\gamma_{ \tiny \mbox{tot}} \triangleq \max(\tilde{\gamma} , \gamma)$ and $w_{ \tiny \mbox{tot}}\triangleq \max (w,\tilde{w})$. This provides:
\begin{align}
     & \| F_G((T_{-t}( \psi(x_{0,a}) \diamond u_a))^-)-F_G((T_{-t}( \psi(x_{0,b})\diamond u_b))^-) \| \\
     &\leq\,\gamma_{ \tiny \mbox{tot}} \left(\max ( \operatorname{ess\,sup}_{s\leq 0}w_{ \tiny \mbox{tot}}(t-s) \| \psi(x_{0,a})(s)-\psi(x_{0,b})(s)\|, \right.  \nonumber \\
     & \left. \operatorname{ess\,sup}_{s\in [0,t]} w_{ \tiny \mbox{tot}}(t-s)  \| u_a(s)-u_b(s)\| \right)\nonumber \\
     &\leq\, \gamma_{ \tiny \mbox{tot}} \left( \operatorname{ess\,sup}_{s\leq 0}w_{ \tiny \mbox{tot}}(-s) \right. \\
     & \left. \| F_G((T_{-t}( \psi(x_{0,a}) \diamond u_a))^-)-F_G((T_{-t}( \psi(x_{0,b})\diamond u_b))^-)\| \right). \nonumber
    \end{align}
Thus $G$ satisfies the uniform FM inequality of Definition~\ref{def:OperatorUniformFM} for all inputs in $\bigcup_{t\geq 0}T_{-t}(\psi(\mathcal{X}_0) \diamond \mathcal{U}^+)$.
\end{proof}

\section{Approximating $\delta$ISS systems using FM} \label{sec:deltaISSFM}

\subsection{Introducing $\delta$ISS for FM} \label{ssec:recapdeltaISS}

If the state-space model~\eqref{eq:StateSpace} has FM {(Definition~\ref{def:FM})}, then its operator realization at $x_0$ has uniform FM {(Definition~\ref{def:OperatorUniformFM})}, so Boyd and Chua's approximation of FM operators holds, and the input-output behavior of \eqref{eq:StateSpace} can be approximated arbitrarily well by a linear time-invariant state-space model followed by a nonlinear static readout~\cite{Boyd1985}. However, practical conditions to prove state-space FM remain to be investigated: to this end, we would like to show that the $\delta$ISS property is sufficient to obtain state-space FM. We recall the definition of a $\delta$ISS system.
\begin{definition}[$\delta$ISS]\label{def:deltaISS}
    The system $\dot{x}(t)=f(x(t),u(t))$ is $\delta$ISS if there exists $\beta \in \mathcal{KL}$ and $\gamma \in \mathcal{K}_{\infty}$ such that for any two initial conditions, $x_a(0),x_b(0) \in \mathcal{X}_0$, considered with two input trajectories, $u_a, u_b\in \mathcal{U}^+$, the corresponding state trajectories satisfy for all $t\geq 0$:
\begin{align} \label{eq:deltaISS_definition}
    \| x_a(t)-x_b(t) \| \leq & \;\beta\big(\| x_a(0)-x_b(0) \|,t\big) \\&+\gamma\big(\operatorname{ess\,sup}_{s\in[0,t]}\| u_a(s)-u_b(s)\| \big). \nonumber
\end{align}
\end{definition}

It is clear, by definition, that input-to-state FM {(Definition~\ref{def:ISSFM})} implies the $\delta$ISS property. However, we would like the opposite implication to hold as well. We will see that, in practice, the converse result holds semi-globally under a physically realistic assumption of equibounded inputs. The strength of this result is that it allows us to use already existing $\delta$ISS criteria to ensure state-space FM semi-globally, such as the following Lyapunov condition:

\begin{theorem}[$\delta$ISS Lyapunov conditions~\cite{Angeli2002ALA}] \label{th:LyapDeltaISS}
     The system $\dot{x}(t)=f(x(t),u(t))$ is globally $\delta$ISS (i.e. $\delta$ISS for $\mathcal{X}_0=\mathbb{R}^{n_x}$) for inputs in $\mathcal{U}^+ = L_{\infty}([0,+\infty),U)$ if there exists a smooth Lyapunov function $V:\mathbb{R}^{n_x} \times \mathbb{R}^{n_x} \to [0,+\infty)$ such that there exists $\alpha_1,\alpha_2,\kappa \in \mathcal{K}_{\infty}$ and $\rho$ positive-definite so for all $x_a,x_b \in \mathbb{R}^{n_x}$ and $u_a,u_b \in U\subseteq \mathbb{R}^{n_y}$:
\begin{subequations} \label{eq:revLyapCond}
\begin{align}
    &\alpha_1(\|x_a-x_b \|) \leq V(x_a,x_b) \leq \alpha_2(\|x_a-x_b \|) , \\
    & \kappa (\| x_a-x_b\|) \geq \|u_a - u_b\| \Rightarrow  \dot{V}(x_a,x_b) \leq -\rho(\|x_a-x_b \|) ,\label{eq:revLyapCond2}
\end{align}
\end{subequations}
where $\dot{V}$ stands for the Lie derivative of $V$ along $f$. Moreover, if the input value set $U$ is compact, {then the inputs are equibounded, and the converse also holds.}
\end{theorem}

\begin{remark} {The original converse result requires extra assumptions, namely convexity of $U$, existence of a fixed point for the vector field, and completeness of solutions~\cite{Angeli2002ALA}. These are not necessary: convexity can be assumed without loss of generality by taking the convex hull of the compact $U$, and~\cite{Kato2025} establishes both the existence of a fixed point and the completeness of solutions purely from the $\delta$ISS inequality. }
\end{remark}

Further sufficient conditions can be found in the literature to demonstrate the $\delta$ISS property, such as contraction-based conditions (Theorem~37 in~\cite{Davydov2022}), or a Killing vector field criterion in~\cite{Giaccagli2023}.

\subsection{$\delta$ISS implies input-to-state FM for equibounded inputs} \label{ssec:deltaISSFMyes}

Under an \textit{equibounded input} assumption, a note from Seung-Jean Kim shows that the $\delta$ISS property implies Boyd and Chua's notion of FM~\cite{Kim}. Since for well-behaved inputs, Boyd and Chua's notion of FM is equivalent to the uniform FM notion for operators introduced in this paper (Proposition~\ref{prop:FM_UniformFM} of Section~\ref{sec:BoydFM}), this result provides a reasonable converse in the operator-theoretic setting.

\begin{theorem}[Well-behaved $\mathcal{U}$: ($\delta$ISS $\Rightarrow$ Uniform FM operator)~\cite{Kim}] \label{th:KimAdapt} If the inputs of $\dot{x}(t)=f(x(t),u(t))$ are equibounded, Lipschitz-continuous, and the system is globally $\delta$ISS, then for all $x_0\in \mathbb{R}^{n_x}$, given an operator $G:\mathcal{U} \to \mathcal{Y}$ such that $G$ is a $(x_0,u_0)$-realization of the system, $G$ has uniform FM {(Definition~\ref{def:OperatorUniformFM})} for all inputs of the form $T_{-t}(u_0 \diamond u)$, with $u\in\mathcal{U^+}$ and $t\geq 0$.
\end{theorem}
\begin{remark}
    {Reference~\cite{Kim} remains unpublished even as a preprint. We have studied it in detail and believe its proof to be correct, but we note that its inaccessibility makes independent verification difficult.}
\end{remark}

We now demonstrate a very similar result for our state-space notion of FM using a Lyapunov converse argument. The result is established using techniques similar in spirit to those employed in the proofs that ISS implies exp-ISS or ISDS~\cite{Praly1996,Grune2002,Karafyllis2018}. However, the incremental nature of the investigated properties imposes restricting the analysis to a compact subset of the state space.

\begin{theorem}[Equibounded $\mathcal{U}^+$, compact $\mathcal{X}_0$: ($\delta$ISS $\Rightarrow$ State-space FM)] \label{th:finalEq} Let $\mathcal{X}_0$ be a compact set of initial conditions. If the inputs of $\dot{x}(t)=f(x(t),u(t))$ are equibounded and the system is globally $\delta$ISS, {then it has input-to-state FM} on $\mathcal{X}_0$ {(Definition~\ref{def:ISSFM})}.
\end{theorem}

\begin{remark}
    Since a FM inequality holds for all compact sets of initial condition, one can say that the system has semi-global input-to-state FM {(Definition~\ref{def:ISSFM})}.
\end{remark}

\begin{proof} Let $\mathcal{X}_0$ be a compact set and let $M_x>0$ be a constant such that for all $x_0 \in \mathcal{X}_0$, $\| x_0 \| \leq M_x$. Let $\mathcal{U}^+$ be equibounded, so there exists $M_u>0$ such that the input signals take value in the compact and convex set $U\triangleq \{u \in \mathbb{R}^{n_u} : \| u\| \leq M_u \} \subset \mathbb{R}^{n_u}$ containing the origin. Since the system is $\delta$ISS, then for $u \triangleq 0$ it is $\delta$GAS, which implies the existence of a unique globally asymptotically stable equilibrium $x^*$~\cite{Kato2025}. Consequently, for all $x_0 \in \mathcal{X}_0$, the $\delta$ISS inequality provides:
\begin{align}
    \| x(t)- x^*\| \leq \beta(M_x + \|x^*\|,t)+\gamma(M_u).
\end{align}
The compact set $\mathcal{X} \triangleq \{x\in \mathbb{R}^{n_x}: \| x- x^*\| \leq \beta(M_x + \|x^*\|,0)+\gamma(M_u)\}$ is such that all state trajectories $x$ of initial condition $x(0) \in \mathcal{X}_0$ remain bounded in $\mathcal{X}$. The converse Lyapunov result of Theorem~\ref{th:LyapDeltaISS} provides the existence of a smooth Lyapunov function $V:\mathbb{R}^{n_x} \times \mathbb{R}^{n_x} \to [0,+\infty)$ such that there exists $\alpha_1,\alpha_2,\kappa \in \mathcal{K}_{\infty}$ and $\rho$ positive-definite satisfying the conditions~\eqref{eq:revLyapCond}~\cite{Angeli2002ALA}. According to Theorem~3.6.10 of~\cite{lakshmikantham1969differential} (see also Lemma~A.34 of~\cite{Praly2022_vol1}), {given a $\lambda>0$, there exists $\alpha_3 \in \mathcal{K}_{\infty}$} such that $W= \alpha_3\circ V$ satisfies: 
\begin{subequations}
\begin{align}
    &\alpha_3 \circ\alpha_1(\|x_a-x_b \|) \leq W(x_a,x_b) \leq \alpha_3 \circ\alpha_2(\|x_a-x_b \|),  \\
    & \kappa (\| x_a-x_b\|) \geq \|u_a - u_b\| \Rightarrow  \dot{W}(x_a,x_b) \leq -\lambda W (x_a,x_b) ,\label{eq:revLyapCond2}
\end{align}
\end{subequations}
with
\begin{equation}
     \dot{W}(x_a,x_b)=\frac{\partial W}{\partial x_a}(x_a,x_b)f(x_a,u_a)+\frac{\partial W}{\partial x_b}(x_a,x_b)f(x_b,u_b).
\end{equation}
Let us consider the set $\mathcal{F}= \{(x_a,x_b,u_a,u_b) :  \| x_a-x_b\| \leq \kappa^{-1} (\|u_a - u_b\|)\} \cap (\mathcal{X}^2 \times U^2)$ and the function $\mu_0$ defined by:
\begin{align}
    &\mu_0(s) \triangleq \\
    &\max_{(x_a,x_b,u_a,u_b) \in \mathcal{F},\| u_a-u_b \|  \leq s } \max(0,\dot{W}(x_a,x_b)+ \lambda W (x_a,x_b) ). \nonumber
\end{align}
{We consider $g(x_a,x_b,u_a,u_b) \triangleq \| u_a-u_b \|$, $\mathcal{P}\triangleq[0,+\infty)$, $h(r)\triangleq r$ and $f(x_a,x_b,u_a,u_b) \triangleq \max(0,\dot{W}(x_a,x_b)+ \lambda W (x_a,x_b) )$. Clearly, $\mathcal{F}$ is closed in $\mathbb{R}^{n_x^2n_u^2}$, $g$ is continuous, $h$ is continuous and strictly increasing, $f$ is continuous, for all $r\geq 0$, $g^{-1}([0,r]) \cap \mathcal{F}$ is compact, $E \triangleq g^{-1}([0,+\infty))$ is not empty, and the local minima of $g$ form the vector space $u_a=u_b$ which is connected in E.} Lemma A.40 of~\cite{Praly2022_vol1} demonstrates that $\mu_0$ is continuous, nondecreasing, and that for all $x_a,x_b \in \mathcal{X}$:
\begin{equation}
    \dot{W}(x_a,x_b) \leq -\lambda W (x_a,x_b)+ \mu(\|u_a - u_b\|),
\end{equation}
with $\mu(r) \triangleq \mu_0(r)+r$, {so we guarantee that $\mu$ is strictly increasing, and thus} $\mu \in \mathcal{K}_{\infty}$. Let $x_a,x_b$ be two state trajectories of initial conditions $x_a(0),x_b(0) \in \mathcal{X}_0$ and of inputs $u_a,u_b \in \mathcal{U}^+$ respectively. Using Grönwall's inequality, this yields:
\begin{align}
    &W(x_a(t),x_b(t))  \\
    & \leq e^{-\lambda t}W(x_a(0),x_b(0))+ \int_0^t e^{-\lambda(t-s)}\mu(\|u_a(s) - u_b(s)\|)ds  \nonumber\\
    &\leq e^{-\lambda t}W(x_a(0),x_b(0)) \\
    &+ \frac{2}{\lambda}\sup_{s\in [0,t]} e^{-\lambda(t-s)/2}\mu(\|u_a(s) - u_b(s)\|) \nonumber
\end{align}
{Now consider $w(t)\triangleq  \mu^{-1}(e^{-\lambda t/2} \mu(2M_u))/2M_u$. Since $e^{-\lambda t/2} \in [0,1]$, then $w(t) \in [0,1]$ as well, and $\lim_{t\to +\infty} w(t) = 0$. Take as well $\kappa(r) \triangleq \mu(\sqrt{2M_u r})$. By composition of $\mathcal{K}_{\infty}$ function, $\kappa \in \mathcal{K}_{\infty}$. Moreover since for all $r\in [0,2 M_u]$ and $t\geq 0$:
\begin{subequations}
    \begin{align}
        &0 \leq \mu^{-1}(e^{-\lambda t/2} \mu(r)) \leq r,\\
        &0 \leq \mu^{-1}(e^{-\lambda t/2} \mu(r)) \leq \mu^{-1}(e^{-\lambda t/2} \mu(2M_u)),
    \end{align}
\end{subequations}
we have:
\begin{equation}
    \mu^{-1}(e^{-\lambda t/2} \mu(r)) \leq \sqrt{ \mu^{-1}(e^{-\lambda t/2} \mu(2M_u)) r},
\end{equation}
and after composition by $\mu$:
\begin{equation}
    e^{-\lambda t/2} \mu(r) \leq \kappa(w(t)r),
\end{equation}
so the following inequality holds:
\begin{equation}
    e^{-\lambda(t-s)/2}\mu(\|u_a(s) - u_b(s)\|) \leq  \kappa\left( w(t-s) \|u_a(s) - u_b(s)\| \right).
\end{equation}}
Consequently:
\begin{align}
    W(x_a(t),x_b(t)) \leq&\, e^{-\lambda t}W(x_a(0),x_b(0)) +  \nonumber \\
    &\frac{2}{\lambda}  \kappa \left( \sup_{s\in [0,t]}w(t-s) \|u_a(s) - u_b(s)\| \right) .
\end{align}
Finally, the following inequalities hold:
\begin{align}
    &\|x_a(t)-x_b(t) \| \leq (\alpha_3 \circ \alpha_1)^{-1}(W(x_a(t),x_b(t))) \\
    &\leq (\alpha_3 \circ \alpha_1)^{-1} \left( \max \left( 2e^{-\lambda t}W(x_a(0),x_b(0)),   \right.\right. \\
    &\left.\left. \frac{4}{\lambda}  \kappa \left( \sup_{s\in [0,t]}w(t-s) \|u_a(s) - u_b(s)\| \right) \right) \right) \nonumber \\
    &\leq  \max \left( (\alpha_3 \circ \alpha_1)^{-1} (2e^{-\lambda t} (\alpha_3 \circ\alpha_2)(\|x_a(0)-x_b(0) \|) ),   \right.    \\
    &\left. (\alpha_3 \circ \alpha_1)^{-1} \left(\frac{4}{\lambda} \kappa \left( \sup_{s\in [0,t]}w(t-s) \|u_a(s) - u_b(s)\| \right)\right)\right),\nonumber
\end{align}
so the FM inequality \eqref{eq:FM_definition_max} is recovered.
\end{proof}

{\begin{remark}
    For a $\delta$ISS system with explicit bounds $M_x$ and $M_u$, this result yields an explicit construction of a memory kernel $w$, although this construction is conservative. Less conservative strategies to obtain $w$ remain to be explored.
\end{remark}}

\begin{remark}
    Some settings where $\delta$ISS does not imply input-to-state FM are discussed in the appendix.
\end{remark}

\subsection{$\delta$ISS representation theorem} \label{ssec:deltaISSapproximation}

\begin{theorem}[$\delta$ISS representation] {If the nonlinear state-space model~\eqref{eq:StateSpace} is globally $\delta$ISS, and has equibounded and Lipschitz-continuous inputs, then, given $x_0\in \mathcal{X}_0$, its input-output behavior $(u(t),y(t))_{t\in [0,+\infty)}$ can be approximated to an arbitrary degree of accuracy by a Wiener model, i.e. by a cascaded system consisting of a linear time-invariant state-space model and a nonlinear static readout. Formally, for all $\varepsilon>0$ there exist $m\in \mathbb{N}$, $A\in \mathbb{R}^{m \times m}$ Hurwitz, $B\in \mathbb{R}^{m \times n_u}$, $z_0 \in \mathbb{R}^m$ and a smooth nonlinear readout $h$ such that for all $u\in \mathcal{U}^+$ and $t\geq 0$:
\begin{equation}
    \left\| y(t)-h\left(z_0+\int_0^t e^{A(t-s)}Bu(s)ds,u(t)\right)\right\| < \varepsilon.
\end{equation}
If Assumption~\ref{ass:operator_statespace3} holds on a compact set $\mathcal{X}_0$ of initial conditions, the approximation applies to all initial conditions $x_0\in \mathcal{X}_0$ by adjusting $z_0$ accordingly.}
\end{theorem}

\begin{proof}
Theorem~\ref{th:finalEq} ensures that the system has input-to-state FM {(Definition~\ref{def:ISSFM})} semi-globally. Moving from the state-space realization to the operator perspective, by Theorem~\ref{th:equivOpeSS}, the operator realization at a fixed initial condition $x_0$ necessarily inherits uniform FM {(Definition~\ref{def:OperatorUniformFM})}. Moreover, under Assumption~\ref{ass:operator_statespace3}, this result can be extended to an entire set of initial conditions using Proposition~\ref{prop:SSFMimpliesUniFM}. The operator realization satisfies the requirements for Boyd and Chua's approximation theorem. This ensures that the input-to-state behavior of~\eqref{eq:StateSpace} can be approximated to an arbitrary degree of accuracy by a cascaded system consisting of a linear time-invariant state-space model and a nonlinear static readout~\cite{Boyd1985}. {The constant $z_0$ is trivially obtained as $ z_0 = \int_{-\infty}^0 e^{-As}B\psi(x_0)(s)ds$.} By composition of nonlinear readouts, the representation can be extended to the entire nonlinear state-space model~\eqref{eq:StateSpace}.
\end{proof}

\begin{remark}{
    Assumption 2 may be by-passable by developing the approximation result directly from the state-space perspective rather than exploiting the operator-theoretic result from Boyd and Chua~\cite{Boyd1985}: this is left as an open question.}
\end{remark}

\section{FM of a current-driven {memristive device}}\label{sec:Mem}

As a final application of the state-space notion of FM studied in this paper, we come back to our motivating example, a current-driven {memristive device}, and we exhibit mild assumptions to ensure that it has FM.

\begin{proposition}[FM of the current-driven {memristive device}] For an equibounded input current $I$, the internal dynamics~\eqref{eq:Memristor2} of the {memristive device} has FM from $I$ to $x$ {(Definition~\ref{def:ISSFM})} on compact sets of initial conditions if it is globally $\delta$ISS. Moreover, if the input current $I$ is continuous and the memristance $M$ is Lipschitz continuous, the {memristive device}~\eqref{eq:Memristor} also has FM from $I$ to $U$ {(Definition~\ref{def:FM})} on compact sets of initial conditions.
\end{proposition}
\begin{proof} 
Let $\overline{I}>0$ and $\lambda>0$ denote the maximum of the signal $I$ and the Lipschitz constant of $M$ respectively. Given that \eqref{eq:Memristor2} has equibounded input currents, and is globally $\delta$ISS, applying Theorem~\ref{th:finalEq} provides the FM of~\eqref{eq:Memristor2} for equibounded input currents $I$ on the compact set of initial conditions $\mathcal{X}_0$. Given a difference of two state trajectories, the following FM inequality holds:
\begin{equation}
     \|\Delta x(t) \| \leq \beta\big(\| \Delta x(0) \|,t\big) +\gamma\big(\operatorname{ess\,sup}_{s\in[0,t]} w(t-s)|\Delta I(s)| \big).
\end{equation}
Given that \eqref{eq:Memristor2} has equibounded input currents, $\delta$ISS also implies for all $\mathcal{X}_0$ compact, there exists a compact $\mathcal{X}$ such that all state trajectories $x(t)$ remain bounded in $\mathcal{X}$, meaning we can also upper bound $\| M(x(t)) \|$ by a constant $\overline{M}$ for all $t\geq 0$. Then, the following inequalities hold:
    \begin{align}
        &|\Delta U(t) | \leq |M(x_a(t))I_a(t)- M(x_b(t))I_b(t) | \\
        &\leq  |M(x_a(t))| \cdot{} |\Delta I(t) | +|I_b(t) | \cdot{}|M(x_a(t))- M(x_b(t)) |  \\
        &\leq \overline{M}|\Delta I(t) |+\overline{I} \lambda \|\Delta x \|.
    \end{align}
 Without loss of generality, we take the memory kernel $w$ with $w(0)=1$, so by continuity of $I$ and $w$, $|\Delta I(t) | \leq \operatorname{ess\,sup}_{s\in[0,t]} w(t-s)|\Delta I(s)|$. Thus, the {memristive device}~\eqref{eq:Memristor} has FM from $I(t)$ to $U(t)$ with $\tilde{\beta} \triangleq \overline{I} \beta$ and $\tilde{\gamma}(r) \triangleq \overline{M}r+\gamma(r)$.
\end{proof}

\section{Conclusion \& perspectives}\label{sec:end}

This work proposes a formulation of the FM property as an extension of the $\delta$IOS property with an explicit memory kernel, with the aim of promoting FM as a core property of real-life nonlinear systems, much like the convolution property of LTI systems. This formulation bridges the operator-theoretic notion of FM with the incremental stability concepts from state-space theory. Although full equivalence with the $\delta$ISS property does not hold in general, we identified reasonable assumptions where FM can be recovered from $\delta$ISS. A current-driven {memristive device} example illustrates the practical relevance of the framework.

{For future work, it would be valuable to derive practical Lyapunov conditions to guarantee global state-space FM, for both the input-to-output and in the input-to-state cases. Moreover, obtaining state-space FM directly from $\delta$IOS is a key open question, worth investigating.} Global extensions of Theorem~\ref{th:KimAdapt} and Theorem~\ref{th:finalEq}, in particular for inputs which are not necessarily Lipschitz-continuous and/or equibounded, remain to be found. Further applications of FM also remain to be explored. In the context of interconnected systems~\cite{Angeli2002ALA}: what can be said about the memory kernel of the interconnection of two FM systems under a small-gain condition? In the context of dissipativity theory~\cite{Sepulchre2022}: could FM help to guarantee incremental passivity properties for memristive systems? It would also be useful to develop a state-space characterization of myopia~\cite{Sandberg1997}, the spatial counterpart of FM, and to investigate how it can be integrated with the proposed state-space notion of FM for the study of partial differential equations. As a final note, the $L_{\infty}$ fading-norm used to define FM can be substituted by a generic $L_p$ fading-norm, which, for $p=1$, creates a bridge between FM and the notion of incremental integral input-to-output stability ($\delta$iIOS) worth investigating~\cite{Angeli2009}.

\bibliographystyle{IEEEtran}
\bibliography{ref}

@article{CHUA2012,
  title = {{Hodgkin-Huxley} axon is made of memristors},
  volume = {22},
  ISSN = {1793-6551},
  website = {http://dx.doi.org/10.1142/S021812741230011X},
  DOI = {10.1142/s021812741230011x},
  number = {03},
  journal = {International Journal of Bifurcation and Chaos},
  publisher = {World Scientific Pub Co Pte Ltd},
  author = {Chua, Leon and Sbitnev, Valery and Kim, Hyongsuk},
  year = {2012},
  month = Mar,
  pages = {1230011}
}

@article{Allan2021,
  title = {Nonlinear Detectability and Incremental Input/Output-to-State Stability},
  volume = {59},
  ISSN = {1095-7138},
  website = {http://dx.doi.org/10.1137/20M135039X},
  DOI = {10.1137/20m135039x},
  number = {4},
  journal = {SIAM Journal on Control and Optimization},
  publisher = {Society for Industrial & Applied Mathematics (SIAM)},
  author = {Allan,  Douglas A. and Rawlings,  James and Teel,  Andrew R.},
  year = {2021},
  month = Jan,
  pages = {3017–3039}
}

@article{Schoukens2017,
  title = {Identification of block-oriented nonlinear systems starting from linear approximations: A survey},
  volume = {85},
  ISSN = {0005-1098},
  website = {http://dx.doi.org/10.1016/j.automatica.2017.06.044},
  DOI = {10.1016/j.automatica.2017.06.044},
  journal = {Automatica},
  publisher = {Elsevier BV},
  author = {Schoukens,  Maarten and Tiels,  Koen},
  year = {2017},
  month = Nov,
  pages = {272–292}
}

@article{Hodgkin1952,
  title = {A quantitative description of membrane current and its application to conduction and excitation in nerve},
  volume = {117},
  ISSN = {1469-7793},
  website = {http://dx.doi.org/10.1113/jphysiol.1952.sp004764},
  DOI = {10.1113/jphysiol.1952.sp004764},
  number = {4},
  journal = {The Journal of Physiology},
  publisher = {Wiley},
  author = {Hodgkin,  A. L. and Huxley,  A. F.},
  year = {1952},
  month = Aug,
  pages = {500–544}
}

@book{Wiener1958,
   title =     {Nonlinear problems in random theory},
   author =    {Norbert Wiener},
   publisher = {MIT Press},
   isbn =      {9780262230049; 0262230046; 9780262730129; 026273012X},
   year =      {1958},
   edition =   {First Edition},
   website =   {libgen.li/file.php?md5=05bff71ae4038169f92416d5082187cb}
}

@article{LOHMILLER1998,
title = {On Contraction Analysis for Non-linear Systems},
journal = {Automatica},
volume = {34},
number = {6},
pages = {683-696},
year = {1998},
issn = {0005-1098},
doi = {https://doi.org/10.1016/S0005-1098(98)00019-3},
website = {https://www.sciencedirect.com/science/article/pii/S0005109898000193},
author = {Winfried Lohmiller and Jean-Jacques E. Slotine},
}

@phdthesis{Burghi2019,
  doi = {10.17863/CAM.57543},
  website = {https://www.repository.cam.ac.uk/handle/1810/310449},
  author = {Burghi,  Thiago},
  language = {en},
  title = {Feedback for neuronal system identification},
  publisher = {Apollo - University of Cambridge Repository},
  year = {2019},
  copyright = {All Rights Reserved}
}

@article{Sontag1989,
  title = {Smooth stabilization implies coprime factorization},
  volume = {34},
  ISSN = {0018-9286},
  website = {http://dx.doi.org/10.1109/9.28018},
  DOI = {10.1109/9.28018},
  number = {4},
  journal = {IEEE Transactions on Automatic Control},
  publisher = {Institute of Electrical and Electronics Engineers (IEEE)},
  author = {Sontag,  E.D.},
  year = {1989},
  month = apr,
  pages = {435–443}
}

@article{Davydov2022,
  title = {Non-Euclidean Contraction Theory for Robust Nonlinear Stability},
  volume = {67},
  ISSN = {2334-3303},
  website = {http://dx.doi.org/10.1109/TAC.2022.3183966},
  DOI = {10.1109/tac.2022.3183966},
  number = {12},
  journal = {IEEE Transactions on Automatic Control},
  publisher = {Institute of Electrical and Electronics Engineers (IEEE)},
  author = {Davydov,  Alexander and Jafarpour,  Saber and Bullo,  Francesco},
  year = {2022},
  month = dec,
  pages = {6667–6681}
}

@inbook{Sontag2008,
  title = {Input to State Stability: Basic Concepts and Results},
  ISBN = {9783540776536},
  ISSN = {0075-8434},
  website = {http://dx.doi.org/10.1007/978-3-540-77653-6_3},
  DOI = {10.1007/978-3-540-77653-6_3},
  booktitle = {Nonlinear and Optimal Control Theory},
  publisher = {Springer Berlin Heidelberg},
  author = {Sontag,  Eduardo D.},
  year = {2008},
  pages = {163–220}
}

@article{Schiller2023,
  title = {On an Integral Variant of Incremental Input/Output-to-State Stability and Its Use as a Notion of Nonlinear Detectability},
  volume = {7},
  ISSN = {2475-1456},
  wbesite = {http://dx.doi.org/10.1109/LCSYS.2023.3286174},
  DOI = {10.1109/lcsys.2023.3286174},
  journal = {IEEE Control Systems Letters},
  publisher = {Institute of Electrical and Electronics Engineers (IEEE)},
  author = {Schiller,  Julian D. and M\"{u}ller,  Matthias A.},
  year = {2023},
  pages = {2341–2346}
}

@article{Sepulchre2022,
  title = {On the incremental form of dissipativity},
  volume = {55},
  ISSN = {2405-8963},
  website = {http://dx.doi.org/10.1016/j.ifacol.2022.11.067},
  DOI = {10.1016/j.ifacol.2022.11.067},
  number = {30},
  journal = {IFAC-PapersOnLine},
  publisher = {Elsevier BV},
  author = {Sepulchre,  Rodolphe and Chaffey,  Thomas and Forni,  Fulvio},
  year = {2022},
  pages = {290–294}
}

@misc{Ortega2024,
  doi = {10.48550/ARXIV.2408.07386},
  website = {https://arxiv.org/abs/2408.07386},
  author = {Ortega,  Juan-Pablo and Rossmannek,  Florian},
  title = {Fading memory and the convolution theorem},
  publisher = {arXiv},
  year = {2024},
  copyright = {arXiv.org perpetual,  non-exclusive license}
}

@BOOK{Praly2022_vol1,
  title    = "Fonctions de Lyapunov, stabilit{\'e}, stabilisation et
              att{\'e}nuation de perturbations: partie 1 : stabilit{\'e}",
  author   = "Praly, Laurent and Bresch-Pietri, Delphine",
  year     =  2022,
  language = "fr"
}

@book{lakshmikantham1969differential,
  title={Differential and integral inequalities: theory and applications: volume I: ordinary differential equations},
  author={Lakshmikantham, Vangipuram and Leela, Srinivasa},
  year={1969},
  publisher={Academic press}
}

@inproceedings{Giaccagli2023,
  title = {Further Results on Incremental Input-to-State Stability Based on Contraction-Metric Analysis},
  website = {http://dx.doi.org/10.1109/CDC49753.2023.10384172},
  DOI = {10.1109/cdc49753.2023.10384172},
  booktitle = {2023 62nd IEEE Conference on Decision and Control (CDC)},
  publisher = {IEEE},
  author = {Giaccagli,  Mattia and Astolfi,  Daniele and Andrieu,  Vincent},
  year = {2023},
  month = dec,
  pages = {1925–1930}
}

@article{Kato2025,
  title = {Incremental global asymptotic stability equals incremental global exponential stability—but at equilibria},
  volume = {59},
  ISSN = {2405-8963},
  website = {http://dx.doi.org/10.1016/j.ifacol.2025.11.071},
  DOI = {10.1016/j.ifacol.2025.11.071},
  number = {19},
  journal = {IFAC-PapersOnLine},
  publisher = {Elsevier BV},
  author = {Kato,  Rui and Astolfi,  Daniele and Andrieu,  Vincent and Praly,  Laurent},
  year = {2025},
  pages = {424–429}
}

@book{day2013thermodynamics,
  title={The thermodynamics of simple materials with fading memory},
  author={Day, William A},
  volume={22},
  year={2013},
  publisher={Springer Science \& Business Media},
  DOI = {10.1007/978-3-642-65318},
}

@article{Kim,
  title = {A Note on the Connection Between Incremental Input-to-State Stability and Fading Memory in Nonlinear Systems},
  journal = {Unpublished},
  author = {Seung-Jean Kim},
  year={2007},
}

@article{Hassani2014,
  title = {A survey on hysteresis modeling,  identification and control},
  volume = {49},
  ISSN = {0888-3270},
  website = {http://dx.doi.org/10.1016/j.ymssp.2014.04.012},
  DOI = {10.1016/j.ymssp.2014.04.012},
  number = {1–2},
  journal = {Mechanical Systems and Signal Processing},
  publisher = {Elsevier BV},
  author = {Hassani,  Vahid and Tjahjowidodo,  Tegoeh and Do,  Thanh Nho},
  year = {2014},
  month = dec,
  pages = {209–233}
}

@article{Angeli2002ALA,
  title={A {Lyapunov} approach to incremental stability properties},
  author={David Angeli},
  journal={IEEE Trans. Autom. Control.},
  year={2002},
  volume={47},
  pages={410-421},
DOI={10.1109/9.989067},
  website={https://api.semanticscholar.org/CorpusID:3184379}
}

@article{Angeli2009,
  title = {Further Results on Incremental Input-to-State Stability},
  volume = {54},
  ISSN = {1558-2523},
  website = {http://dx.doi.org/10.1109/TAC.2009.2015561},
  DOI = {10.1109/tac.2009.2015561},
  number = {6},
  journal = {IEEE Transactions on Automatic Control},
  publisher = {Institute of Electrical and Electronics Engineers (IEEE)},
  author = {Angeli,  D.},
  year = {2009},
  month = jun,
  pages = {1386–1391}
}

@article{Kellett2014,
  title = {A compendium of comparison function results},
  volume = {26},
  ISSN = {1435-568X},
  website = {http://dx.doi.org/10.1007/s00498-014-0128-8},
  DOI = {10.1007/s00498-014-0128-8},
  number = {3},
  journal = {Mathematics of Control,  Signals,  and Systems},
  publisher = {Springer Science and Business Media LLC},
  author = {Kellett,  Christopher M.},
  year = {2014},
  month = mar,
  pages = {339–374}
}

@article{Coleman1966,
  title = {Norms and semi-groups in the theory of fading memory},
  volume = {23},
  ISSN = {1432-0673},
  website =  {http://dx.doi.org/10.1007/BF00251727},
  DOI = {10.1007/bf00251727},
  number = {2},
  journal = {Archive for Rational Mechanics and Analysis},
  publisher = {Springer Science and Business Media LLC},
  author = {Coleman,  Bernard D. and Mizel,  Victor J.},
  year = {1966},
  month = jan,
  pages = {87–123}
}

@article{ilhav2019,
  title = {The scientific work of {Bernard D. Coleman}},
  volume = {24},
  ISSN = {1741-3028},
  website =  {http://dx.doi.org/10.1177/1081286519825544},
  DOI = {10.1177/1081286519825544},
  number = {2},
  journal = {Mathematics and Mechanics of Solids},
  publisher = {SAGE Publications},
  author = {Šilhavý,  M},
  year = {2019},
  month = feb,
  pages = {339–360}
}

@article{Coleman1960,
  title = {An approximation theorem for functionals,  with applications in continuum mechanics},
  volume = {6},
  ISSN = {1432-0673},
  website =  {http://dx.doi.org/10.1007/BF00276168},
  DOI = {10.1007/bf00276168},
  number = {1},
  journal = {Archive for Rational Mechanics and Analysis},
  publisher = {Springer Science and Business Media LLC},
  author = {Coleman,  Bernard D. and Noll,  Walter},
  year = {1960},
  month = jan,
  pages = {355–370}
}

@article{Zang2003,
  title = {Fading memory and stability},
  volume = {340},
  ISSN = {0016-0032},
  website =  {http://dx.doi.org/10.1016/j.jfranklin.2003.11.002},
  DOI = {10.1016/j.jfranklin.2003.11.002},
  number = {6–7},
  journal = {Journal of the Franklin Institute},
  publisher = {Elsevier BV},
  author = {Zang,  Guoqiang and Iglesias,  Pablo A.},
  year = {2003},
  month = sep,
  pages = {489–502}
}

@ARTICLE{Dambre2012,
  title     = "Information processing capacity of dynamical systems",
  author    = "Dambre, Joni and Verstraeten, David and Schrauwen, Benjamin and
               Massar, Serge",
  journal   = "Sci. Rep.",
  publisher = "Springer Science and Business Media LLC",
  volume    =  2,
  number    =  1,
  pages     = "514",
  month     =  jul,
  year      =  2012,
  copyright = "https://creativecommons.org/licenses/by-nc-sa/3.0/",
  language  = "en",
  DOI = "10.1038/srep00514",
}

@article{Dahleh1995,
  title = {Worst-case identification of nonlinear fading memory systems},
  volume = {31},
  ISSN = {0005-1098},
  website =  {http://dx.doi.org/10.1016/0005-1098(94)00131-2},
  DOI = {10.1016/0005-1098(94)00131-2},
  number = {3},
  journal = {Automatica},
  publisher = {Elsevier BV},
  author = {Dahleh,  Munther A. and Sontag,  Eduardo D. and Tse,  David N.C. and Tsitsiklis,  John N.},
  year = {1995},
  month = mar,
  pages = {503–508}
}

@article{Partington1996,
  title = {Modeling of linear fading memory systems},
  volume = {41},
  ISSN = {0018-9286},
  website =  {http://dx.doi.org/10.1109/9.506247},
  DOI = {10.1109/9.506247},
  number = {6},
  journal = {IEEE Transactions on Automatic Control},
  publisher = {Institute of Electrical and Electronics Engineers (IEEE)},
  author = {Partington,  J.R. and Makila,  P.M.},
  year = {1996},
  month = jun,
  pages = {899–903}
}

@article{Shamma1993,
  title = {Fading-memory feedback systems and robust stability},
  volume = {29},
  ISSN = {0005-1098},
  website =  {http://dx.doi.org/10.1016/0005-1098(93)90182-S},
  DOI = {10.1016/0005-1098(93)90182-s},
  number = {1},
  journal = {Automatica},
  publisher = {Elsevier BV},
  author = {Shamma,  Jeff S. and Zhao,  Rongze},
  year = {1993},
  month = jan,
  pages = {191–200}
}

@article{Shamma1991,
  title = {The necessity of the small-gain theorem for time-varying and nonlinear systems},
  volume = {36},
  ISSN = {0018-9286},
  website =  {http://dx.doi.org/10.1109/9.90227},
  DOI = {10.1109/9.90227},
  number = {10},
  journal = {IEEE Transactions on Automatic Control},
  publisher = {Institute of Electrical and Electronics Engineers (IEEE)},
  author = {Shamma,  Jeff S.},
  year = {1991},
  pages = {1138–1147}
}

@misc{Ortega2025,
  doi = {10.48550/ARXIV.2508.19145},
  website =  {https://arxiv.org/abs/2508.19145},
  author = {Ortega,  Juan-Pablo and Rossmannek,  Florian},
  title = {Echoes of the past: A unified perspective on fading memory and echo states},
  publisher = {arXiv},
  year = {2025},
  copyright = {arXiv.org perpetual,  non-exclusive license}
}

@article{Gonon2021,
  title = {Fading memory echo state networks are universal},
  volume = {138},
  ISSN = {0893-6080},
  website =  {http://dx.doi.org/10.1016/j.neunet.2021.01.025},
  DOI = {10.1016/j.neunet.2021.01.025},
  journal = {Neural Networks},
  publisher = {Elsevier BV},
  author = {Gonon,  Lukas and Ortega,  Juan-Pablo},
  year = {2021},
  month = jun,
  pages = {10–13}
}

@article{Grune2002,
  title = {Input-to-state dynamical stability and its {Lyapunov} function characterization},
  volume = {47},
  ISSN = {0018-9286},
  website =  {http://dx.doi.org/10.1109/TAC.2002.802761},
  DOI = {10.1109/tac.2002.802761},
  number = {9},
  journal = {IEEE Transactions on Automatic Control},
  publisher = {Institute of Electrical and Electronics Engineers (IEEE)},
  author = {Grune,  L.},
  year = {2002},
  month = sep,
  pages = {1499–1504}
}

@article{Praly1996,
  title = {Stabilization in spite of matched unmodeled dynamics and an equivalent definition of input-to-state stability},
  volume = {9},
  ISSN = {1435-568X},
  website =  {http://dx.doi.org/10.1007/BF01211516},
  DOI = {10.1007/bf01211516},
  number = {1},
  journal = {Mathematics of Control,  Signals,  and Systems},
  publisher = {Springer Science and Business Media LLC},
  author = {Praly,  Laurent and Wang,  Yuan},
  year = {1996},
  month = mar,
  pages = {1–33}
}

@inbook{Karafyllis2018,
  title = {Fading Memory Input-to-State Stability},
  ISBN = {9783319910116},
  ISSN = {2197-7119},
  website =  {http://dx.doi.org/10.1007/978-3-319-91011-6_7},
  DOI = {10.1007/978-3-319-91011-6_7},
  booktitle = {Input-to-State Stability for PDEs},
  publisher = {Springer International Publishing},
  author = {Karafyllis,  Iasson and Krstic,  Miroslav},
  year = {2018},
  month = jun,
  pages = {185–191}
}

@article{Coleman1968,
  title = {On the general theory of fading memory},
  volume = {29},
  ISSN = {1432-0673},
  website =  {http://dx.doi.org/10.1007/BF00256456},
  DOI = {10.1007/bf00256456},
  number = {1},
  journal = {Archive for Rational Mechanics and Analysis},
  publisher = {Springer Science and Business Media LLC},
  author = {Coleman,  Bernard D. and Mizel,  Victor J.},
  year = {1968},
  month = jan,
  pages = {18–31}
}

@article{Sandberg2002R,
  title = {$\mathbf{R}_+$ fading memory and extensions of input-output maps},
  volume = {49},
  ISSN = {1057-7122},
  website =  {http://dx.doi.org/10.1109/TCSI.2002.804547},
  DOI = {10.1109/tcsi.2002.804547},
  number = {11},
  journal = {IEEE Transactions on Circuits and Systems I: Fundamental Theory and Applications},
  publisher = {Institute of Electrical and Electronics Engineers (IEEE)},
  author = {Sandberg,  I. W.},
  year = {2002},
  month = nov,
  pages = {1586–1591}
}

@article{Sandberg2001Z,
  title = {$\mathbf{Z}_+$ fading memory and extensions of input–output maps},
  volume = {29},
  ISSN = {1097-007X},
  website =  {http://dx.doi.org/10.1002/cta.156},
  DOI = {10.1002/cta.156},
  number = {4},
  journal = {International Journal of Circuit Theory and Applications},
  publisher = {Wiley},
  author = {Sandberg,  Irwin W.},
  year = {2001},
  month = jul,
  pages = {381–388}
}

@article{Sandberg2003,
  title = {Notes on Fading-Memory Conditions},
  volume = {22},
  ISSN = {1531-5878},
  website =  {http://dx.doi.org/10.1007/s00034-004-7012-6},
  DOI = {10.1007/s00034-004-7012-6},
  number = {1},
  journal = {Circuits,  Systems,  and Signal Processing},
  publisher = {Springer Science and Business Media LLC},
  author = {Sandberg,  Irwin W.},
  year = {2003},
  month = jan,
  pages = {43–55}
}

@article{Matthews1993,
  title = {Approximating nonlinear fading-memory operators using neural network models},
  volume = {12},
  ISSN = {1531-5878},
  website =  {http://dx.doi.org/10.1007/BF01189878},
  DOI = {10.1007/bf01189878},
  number = {2},
  journal = {Circuits,  Systems,  and Signal Processing},
  publisher = {Springer Science and Business Media LLC},
  author = {Matthews,  Michael B.},
  year = {1993},
  month = jun,
  pages = {279–307}
}

@article{Sugiura2024,
  title = {Nonessentiality of Reservoir’s Fading Memory for Universality of Reservoir Computing},
  volume = {35},
  ISSN = {2162-2388},
  website =  {http://dx.doi.org/10.1109/TNNLS.2023.3298013},
  DOI = {10.1109/tnnls.2023.3298013},
  number = {11},
  journal = {IEEE Transactions on Neural Networks and Learning Systems},
  publisher = {Institute of Electrical and Electronics Engineers (IEEE)},
  author = {Sugiura,  Shuhei and Ariizumi,  Ryo and Asai,  Toru and Azuma,  Shun-Ichi},
  year = {2024},
  month = nov,
  pages = {16801–16815}
}

@misc{Forni2025,
      title={Gradient modelling of memristive systems}, 
      author={Fulvio Forni and Rodolphe Sepulchre},
      year={2025},
      eprint={2504.10093},
      archivePrefix={arXiv},
      primaryClass={eess.SY},
      website={https://arxiv.org/abs/2504.10093}, 
}

@misc{Donchev2025,
  doi = {10.48550/ARXIV.2504.06009},
  website =  {https://arxiv.org/abs/2504.06009},
  author = {Donchev,  Tihol Ivanov and Shali,  Brayan M. and Sepulchre,  Rodolphe},
  title = {Linear time-and-space-invariant relaxation systems},
  publisher = {arXiv},
  year = {2025},
  copyright = {Creative Commons Attribution 4.0 International}
}

@article{Sandberg1997,
  title = {Uniform approximation of multidimensional myopic maps},
  volume = {44},
  ISSN = {1057-7122},
  website =  {http://dx.doi.org/10.1109/81.585959},
  DOI = {10.1109/81.585959},
  number = {6},
  journal = {IEEE Transactions on Circuits and Systems I: Fundamental Theory and Applications},
  publisher = {Institute of Electrical and Electronics Engineers (IEEE)},
  author = {Sandberg,  I.W. and Xu,  L.},
  year = {1997},
  month = jun,
  pages = {477–500}
}

@misc{Huo2024,
  doi = {10.48550/ARXIV.2403.11945},
  website =  {https://arxiv.org/abs/2403.11945},
  author = {Huo,  Yongkang and Chaffey,  Thomas and Sepulchre,  Rodolphe},
  title = {Kernel Modelling of Fading Memory Systems},
  publisher = {arXiv},
  year = {2024},
  copyright = {Creative Commons Attribution Share Alike 4.0 International}
}

@article{Boyd1985,
  title = {Fading memory and the problem of approximating nonlinear operators with {Volterra} series},
  volume = {32},
  ISSN = {0098-4094},
  website =  {http://dx.doi.org/10.1109/TCS.1985.1085649},
  DOI = {10.1109/tcs.1985.1085649},
  number = {11},
  journal = {IEEE Transactions on Circuits and Systems},
  publisher = {Institute of Electrical and Electronics Engineers (IEEE)},
  author = {Boyd,  S. and Chua,  L. O.},
  year = {1985},
  month = nov,
  pages = {1150–1161}
}

@article{Sepulchre2021,
  title = {Fading Memory [From the Editor]},
  volume = {41},
  ISSN = {1941-000X},
  website =  {http://dx.doi.org/10.1109/MCS.2020.3033098},
  DOI = {10.1109/mcs.2020.3033098},
  number = {1},
  journal = {IEEE Control Systems},
  publisher = {Institute of Electrical and Electronics Engineers (IEEE)},
  author = {Sepulchre,  Rodolphe},
  year = {2021},
  month = feb,
  pages = {4–5}
}

@article{Chua1976,
  title = {Memristive devices and systems},
  volume = {64},
  ISSN = {0018-9219},
  website =  {http://dx.doi.org/10.1109/PROC.1976.10092},
  DOI = {10.1109/proc.1976.10092},
  number = {2},
  journal = {Proceedings of the IEEE},
  publisher = {Institute of Electrical and Electronics Engineers (IEEE)},
  author = {Chua,  L. O. and Sung Mo Kang},
  year = {1976},
  pages = {209–223}
}

\begin{IEEEbiography}
[{\includegraphics[width=1in,height=1.25in,clip,keepaspectratio]{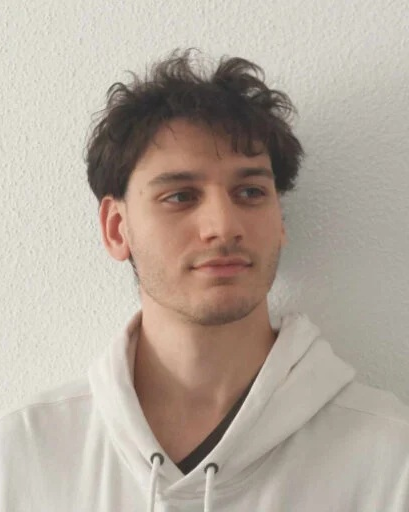}}]{Gustave Bainier} received an Engineering Master’s degree from École Centrale Nantes, France, in 2020, and a Ph.D. degree in Automatic Control from Université de Lorraine, France, in 2024. His doctoral research, carried out at CRAN, CNRS in Nancy, France, under the supervision of J.-C. Ponsart and B. Marx, focused on theoretical developments in the Linear Parameter Varying (LPV) and Takagi–Sugeno (T-S) control frameworks, with applications to fault diagnosis. He is currently a Post-Doctoral Researcher with the Neuroengineering Laboratory at University of Liège, Belgium, where his research explores applications of control theory to brain-inspired computing. His research interests include neuroengineering, LPV systems, and nonlinear stability analysis.
\end{IEEEbiography}

\begin{IEEEbiography}
[{\includegraphics[width=1in,height=1.25in,clip,keepaspectratio]{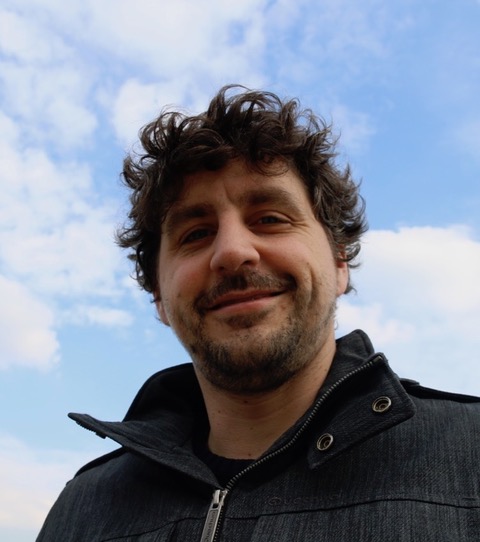}}]{Antoine Chaillet} is a Professor at Laboratoire des Signaux et Systèmes (L2S), CentraleSupélec, Paris Saclay University, France. He received his PhD from Université Paris Sud, France in 2005. He was a post-doctoral fellow at Università di Pisa, Italy. He was a junior member of Institut Universitaire de France from 2016 to 2021. He served as an as-sociate editor for IEEE Transactions on Automatic Control from 2019 to 2024. He is now director of the H-CODE institute of Paris Saclay University. His research interests include stability and robustness analysis of nonlinear systems, time-delay systems, and neuroscience applications. He his the author of around 100 peer-reviewed papers on these topics.
\end{IEEEbiography}

\begin{IEEEbiography}[{\includegraphics[width=1in,height=1.25in,clip,keepaspectratio]{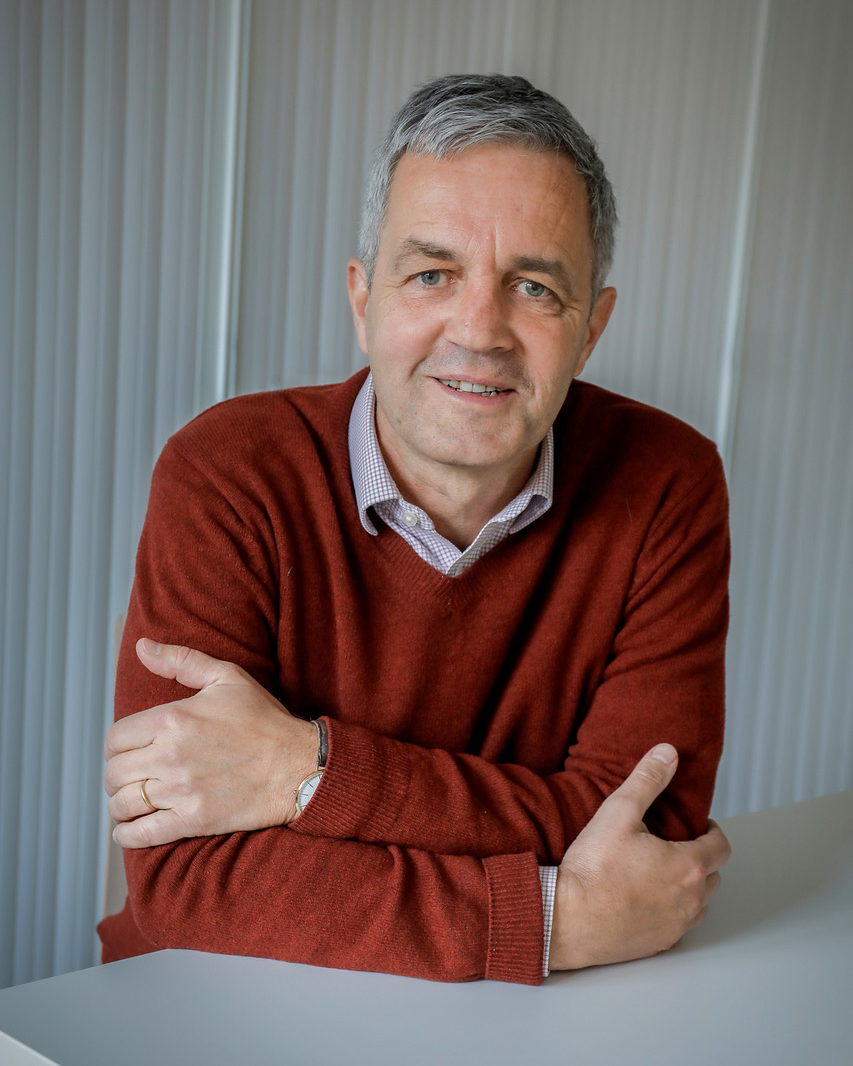}}]{Rodolphe Sepulchre} is professor of engineering at the KU Leuven (Belgium) and at the University of  Cambridge (UK). He is a fellow of EUCA (2026), IFAC (2020), IEEE (2009), and SIAM (2015). He received the IEEE CSS Antonio Ruberti Young Researcher Prize in 2008 and the IEEE CSS George S. Axelby Outstanding Paper Award in 2020. He was elected at the Royal Academy of Belgium in 2013.  He co-authored the monographs Constructive Nonlinear Control (1997, with M. Jankovic and P. Kokotovic) and Optimization on Matrix Manifolds (2008, with P.-A. Absil and R. Mahony). His current research interests are in neuromorphic control.
\end{IEEEbiography}

\begin{IEEEbiography}
[{\includegraphics[width=1in,height=1.25in,clip,keepaspectratio]{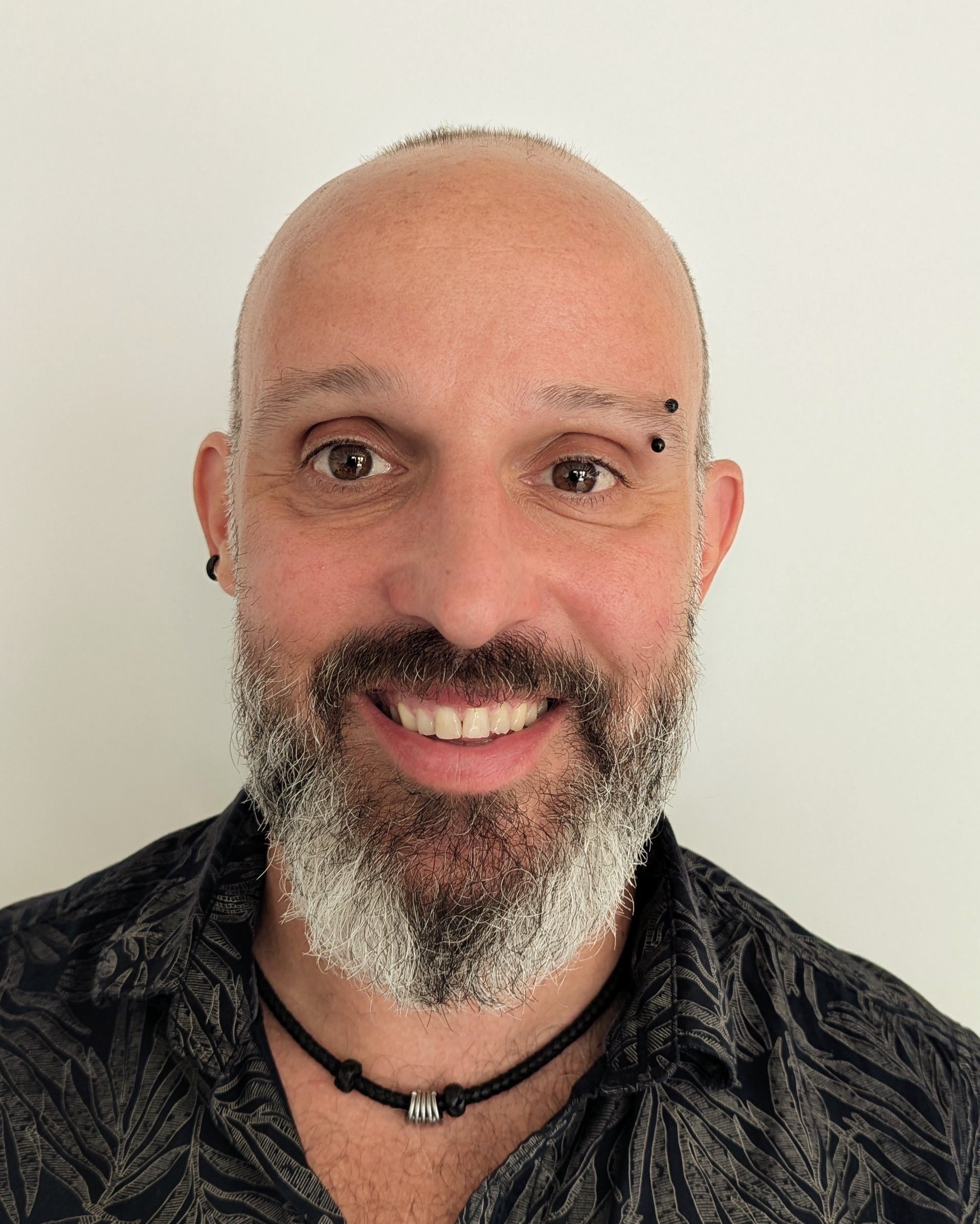}}]{Alessio Franci}  is a professor in the Montefiore Institute of the University of Liege and a grantee of the WEL Research Institute. He is a recipient of the 2025 IEEE CSS George S. Axelby Outstanding Paper Award. He got his MSc from the University of Pisa in 2008 and his PhD from the University of Paris Sud XI in 2012. He was a postdoctoral fellow at the University of Liege, INRIA Lille, and University of Cambridge. Between 2015 and 2022 he was a professor in the Math Department of the National Autonomous University of Mexico. His research interests are in neuromorphic control theory and spiking control systems, the design of excitable circuits, sensors, and actuators, and applications in neuroscience and robotics.
\end{IEEEbiography}

\appendix
\subsection{Some settings where $\delta$ISS does not imply input-to-state FM}\label{app:deltaISSFMno}

Several counterexamples demonstrate that the implication fails in the presence of delays or time-varying parameters. The discussion becomes particularly nuanced for time-invariant, delay-free systems, as we were unable to construct a globally valid counterexample, and can therefore only refute the implication by restricting the set of initial conditions. In the absence of delay, the contradictions are obtained via an unbounded sequence of input functions, causing the reverse implication to fail from a purely mathematical standpoint.

Let us begin by showing that $\delta$ISS does not imply FM when arbitrarily large input delays are allowed.

\begin{counterexample}[Input delay: ($\delta$ISS $\not \Rightarrow$ FM)] \label{cex:1}
    The following system is $\delta$ISS (at $t_0=0$), but it does not have input-to-state FM (at $t_0=0$). 
    \begin{equation} \label{eq:counterexample1}
    \dot{x}(t)=-x(t)+\int_0^{\min(1,t)} u(s)ds.
\end{equation}
\end{counterexample}
\begin{proof} By linearity, the difference between two trajectories satisfies the following inequalities:
\begin{align}
    | \Delta x(t) |& \leq  e^{-t} |\Delta x(0)| +\int_0^t e^{s-t} \int_0^{\min(1,s)} |\Delta u(r) |dr \,ds \\
    & \leq e^{-t} |\Delta x(0)| +\int_0^t e^{s-t} \,\operatorname{ess\,sup}_{r\in[0,1]} |\Delta u(r) |ds \\
    &  \leq e^{-t} |\Delta x(0)| +  \operatorname{ess\,sup}_{s\in[0,t]} |\Delta u(s) | ,
\end{align}
hence the system is $\delta$ISS at $t_0=0$. However, consider $\Delta u(t)=1$ for $t\leq 1$ and $\Delta u(t)=0$ for $t>1$. {The incremental dynamics yield $\Delta \dot{x}(t)=-\Delta x(t)+1$ for all $t>1$, so $\lim_{t\to +\infty} \Delta x(t) = 1$}, and the CICO property of FM is not verified (Proposition~\ref{prop:CICO}), hence the system does not have FM.
\end{proof}
\begin{remark}
    Similar infinite memory effects may arise in practice, such as in hysteresis phenomena~\cite{Hassani2014}.
\end{remark}
\begin{remark}
    The delay of this counterexample can be removed using a time-varying integrator:
    \begin{subequations}
        \begin{align}
            \dot{x}_1(t)&=1_{\leq 1}(t)u(t)\\
            \dot{x}_2(t)&=-x_2(t)+x_1(t) 
        \end{align}
    \end{subequations}
    with $1_{\leq 1}(t) \triangleq 1$ if $t\leq 1$, and $1_{\leq 1}(t) \triangleq 0$ else. However, this state-augmented system is not $\delta$ISS anymore.
\end{remark}

More surprisingly, even in the absence of delays, time-varying parameters are sufficient to break the implication. The intuition behind the following counterexample is based on the following averaging I/O map:
\begin{equation} \label{eq:counterexample2_pre}
    y(t)=\frac{1}{t+1}\int_0^t u(s)ds.
\end{equation}
Although the relative influence of past inputs $u$ on the output $y$ diminishes over time, this decay occurs regardless of how far in the past the inputs were applied, and crucially, it is also vanishing linearly. This averaging process exhibits fragile stability, making the system lack a decaying memory mechanism, i.e. FM. By differentiation, \eqref{eq:counterexample2_pre} admits a time-varying and delay-free state-space representation.
\begin{counterexample}[Time-varying: ($\delta$ISS $\not \Rightarrow$ FM)] \label{cex:2}
    The following system is $\delta$ISS (at $t_0=0$), but it does not have input-to-state FM (at $t_0=0$).
    \begin{equation}\label{eq:counterexample2}
        \dot{x}(t)=\frac{1}{t+1}\left(-x(t) +u(t)\right).
    \end{equation} 
\end{counterexample}

\begin{proof} By linearity, the difference between two trajectories of \eqref{eq:counterexample2} is easily determined:
\begin{align}
\Delta x(t) &= e^{-\int_0^t \frac{ds}{s+1}}\Delta x(0)+\int_0^t e^{-\int_s^t \frac{dr}{r+1}}\frac{\Delta u(s)}{s+1}ds\\
&=e^{-\ln(t+1)}\Delta x(0)+\int_0^t e^{\cancel{\ln(s+1)}-\ln(t+1)}\frac{\Delta u(s)}{\cancel{s+1}}ds \\
&=\frac{1}{t+1}\left(\Delta x(0)+\int_0^t \Delta u(s)ds\right),
\end{align}
so the following inequality holds, demonstrating that the system is $\delta$ISS at $t_0=0$:
\begin{equation}
    |\Delta x(t)| \leq \frac{1}{t+1}|\Delta x(0)| + \underbrace{\frac{t}{t+1}}_{\leq 1}\operatorname{ess\,sup}_{s\in[0,t]} |\Delta u(s) | .
\end{equation}
The claim that \eqref{eq:counterexample2} does not have input-to-state FM at $t_0=0$ is shown by contradiction. Assuming that \eqref{eq:counterexample2} has input-to-state FM at $t_0=0$, there exists $\gamma \in \mathcal{K}_{\infty}$ and a memory kernel $w:[0,+\infty)\to(0,1]$ satisfying Definition~\ref{def:memorykernel} ($0$ can be excluded thanks to Proposition~\ref{prop:MemoryKernelComp}) such that for all $\Delta u \in L_{\infty}$ and $\Delta x(0)=0$:
\begin{align} \label{eq:ineqProofCounterExample2}
    |\Delta x(t)| = &\frac{\left| \int_0^t \Delta u(s)ds\right|}{t+1} \leq  \gamma \left( \operatorname{ess\,sup}_{s\in[0,t]} w(t-s)|\Delta u(s) | \right).
\end{align}
Let us show that $(t+1)^{-1}\int_0^t w(s)^{-1}ds$ diverges as $t\to+\infty$. Let $M>0$. It is shown that there exists $t^*>0$ such that for all $t\geq t^*$, $(t+1)^{-1}\int_0^t w(s)^{-1}ds \geq M$. Since $w(t)$ is nonincreasing and $\lim_{t\to +\infty} w(t) =0^+$, $w(t)^{-1}$ is nondecreasing and $\lim_{t\to +\infty} w(t)^{-1} =+\infty$, hence there exists $t_1>0$ such that for all $t\geq t_1$, $w(t)^{-1} \geq 2M$. Take $t^*=2t_1+1$. The following inequalities hold for all $t\geq t^*$:
\begin{align}
    \frac{1}{t+1} \int_0^t\frac{ds }{ w(s)} &\geq \frac{1}{t+1} \int_{t_1}^{t}\frac{ds }{ w(s)} \geq \frac{t-t_1}{t+1}2M \geq M  ,
\end{align}
thus demonstrating the divergence of $(t+1)^{-1}\int_0^t w(s)^{-1}ds$.
In particular, this demonstrates that $|\Delta x(T)|$ diverges as $T\to +\infty$ for inputs defined by $\Delta u(t)\triangleq w(T-t)^{-1}$ for $t\in [0,T]$. However, the FM inequality~\eqref{eq:ineqProofCounterExample2} yields:
\begin{equation}
    |\Delta x(T)| \leq \gamma \left( \operatorname{ess\,sup}_{s\in[0,T]} w(T-s) w(T-s)^{-1} \right) = \gamma(1),
\end{equation}
meaning $\gamma(1)\to +\infty$, a contradiction.
\end{proof}
\begin{remark}
    In particular, this system provides a time-varying counterexample to the claim that a CICO property is sufficient to guarantee input-to-state FM for a $\delta$ISS system, which would otherwise appear to be a natural converse of Proposition~\ref{prop:CICO}.
\end{remark}
Now what if the system is both delay-free and time-invariant, i.e. of the form $\dot{x}(t)=f(x(t),u(t))$? A variation of the previous counterexample shows that restricting $\mathcal{X}_0$ is sufficient to obtain a $\delta$ISS system without FM in this context as well.
\begin{counterexample}[Restricting $\mathcal{X}_0$: ($\delta$ISS $\not \Rightarrow$ FM)] The following system is $\delta$ISS for $\mathcal{X}_0=\{1\}\times \mathbb{R}$, but it does not have input-to-state FM for $\mathcal{X}_0=\{1\}\times \mathbb{R}$.\begin{subequations}
        \begin{align}
            \dot{x}_1(t)&=-x_1(t)^2, \\
            \dot{x}_2(t)&=x_1(t)(-x_2(t)+u(t)).
        \end{align}
    \end{subequations}
\end{counterexample}
\begin{proof}
    If $x_1(0)=1$, then $x_1(t)=\frac{1}{t+1}$, and the rest of the proof follows from the previous counterexample.
\end{proof}
We were unable to find a counterexample holding globally ($\mathcal{X}_0=\mathbb{R}^{n_x}$). Thus, the question of whether global $\delta$ISS implies global input-to-state FM for time-invariant, delay-free systems remains open.
\end{document}